\documentclass[11pt]{article}

\usepackage{macros}

\newcommand{\parhead}[1]{\medskip \noindent {\bfseries\boldmath\ignorespaces {#1}}\hskip 0.9em}
\newcommand{\shift}{\mathsf{Shift}}
\newcommand{\Permutation}[1]{\wh{\calS}_{#1}}
\newcommand{\edgefactor}{\textsf{edge-factor}}
\newcommand{\Decode}{\mathsf{Decode}}
\newcommand{\FAIL}{\mathsf{FAIL}}
\newcommand{\into}{\mathrm{in}}
\newcommand{\out}{\mathrm{out}}
\newcommand{\nBR}{\operatorname{nBR}}

\newcommand{\Fresh}{\mathsf{F}}
\newcommand{\NI}{\mathsf{N}}
\newcommand{\Ret}{\mathsf{R}}
\newcommand{\Highmult}{\mathsf{H}}
\newcommand{\Paired}{\mathsf{P}}
\newcommand{\fresh}{\mathfrak{f}}
\newcommand{\n}{\mathfrak{n}}
\newcommand{\ret}{\mathfrak{r}}
\newcommand{\highmult}{\mathfrak{h}}
\newcommand{\paired}{\mathfrak{p}}

\begin{document}

\title{Certifying Euclidean Sections and Finding Planted Sparse Vectors Beyond the $\sqrt{n}$ Dimension Threshold}

\author{Venkatesan Guruswami\thanks{Simons Institute for the Theory of Computing, and Departments of EECS and Mathematics, UC Berkeley.  Email: \url{venkatg@berkeley.edu}. Research supported in part by a Simons Investigator Award and NSF grant CCF-2211972.}
\and Jun-Ting Hsieh\thanks{Carnegie Mellon University. Email: \url{juntingh@cs.cmu.edu}. Supported by NSF CAREER Award \#2047933.}
\and Prasad Raghavendra\thanks{UC Berkeley. Email: \url{raghavendra@berkeley.edu}.
Supported by NSF CCF-2342192.}
}
\date{April 2024}

\maketitle

\begin{abstract}
    We consider the task of certifying that a random $d$-dimensional subspace $X$ in $\R^n$ is well-spread --- every vector $x \in X$ satisfies $c\sqrt{n} \|x\|_2 \leq \|x\|_1 \leq \sqrt{n}\|x\|_2$.
    In a seminal work, Barak et.~al.~\cite{BarakBHKSZ12} showed a polynomial-time certification algorithm when $d \leq O(\sqrt{n})$.
    On the other hand, when $d \gg \sqrt{n}$, the certification task is information-theoretically possible but there is evidence that it is computationally hard~\cite{MaoW21,ChendO22}, a phenomenon known as the information-computation gap.

    In this paper, we give subexponential-time certification algorithms in the $d \gg \sqrt{n}$ regime.
    Our algorithm runs in time $\exp(\wt{O}(n^{\eps}))$ when $d \leq \wt{O}(n^{\frac{1+\eps}{2}})$, establishing a smooth trade-off between runtime and the dimension.

    Our techniques naturally extend to the related planted problem, where the task is to recover a sparse vector planted in a random subspace.
    Our algorithm achieves the same runtime and dimension trade-off for this task.
\end{abstract}

\thispagestyle{empty}
\setcounter{page}{0}
\newpage

\section{Introduction}

For any vector $x \in \R^n$, we know that $\|x\|_2 \leq \|x\|_1 \leq \sqrt{n} \|x\|_2$.
Intuitively, if the upper bound is approximately tight, i.e., $\|x\|_1 \geq c \sqrt{n} \|x\|_2$, then $x$ is ``dense'' or ``incompressible''.
We say that a subspace $X \subseteq \R^n$ is \emph{well-spread} if every $x\in X$ is dense.
More formally, we define the \emph{distortion} of $X$, denoted $\Delta(X)$, as follows,
\begin{equation*}
    \Delta(X) \coloneqq \sup_{x\in X:\ x\neq 0} \frac{\sqrt{n} \|x\|_2}{\|x\|_1} \mper
    \numberthis \label{eq:distortion}
\end{equation*}
The distortion always satisfies $1 \leq \Delta(X) \leq \sqrt{n}$.
Subspaces with large dimension $d$ and small distortion are called good \emph{Euclidean sections} of $\ell_1^n$.
In particular, they provide embeddings of $\ell_2^d$ into $\ell_1^n$ with only a small blow-up in the dimension, thus such subspaces have many applications including error-correcting codes over the reals~\cite{CandesT05,GuruswamiLW08}, compressed sensing~\cite{KashinT07,Donoho06}, high-dimensional nearest-neighbor search~\cite{Indyk06}, and oblivious regression~\cite{dOrsiLNSS21}.

Therefore, there has been a long line of work on constructing good Euclidean sections, including explicit constructions~\cite{Indyk06,Indyk07,GuruswamiLR10} and constructions requiring few random bits (a.k.a.\ partially derandomized constructions) \cite{LovettS08,GuruswamiLW08,IndykS10}.
However, all such constructions suffer in either the dimension or the distortion (see \cite{IndykS10} and references therein).

On the other hand, it is well known that a fully random subspace of $\R^n$ of dimension $d = \Omega(n)$ (for instance, the column space of an $n \times d$ matrix with i.i.d.\ Gaussian or $\pm 1$ entries) has $O(1)$ distortion with high probability \cite{FigielLM77,Kashin77,GarnaevG84}.
However, randomized constructions suffer from the drawback that the output is not guaranteed to \emph{always} satisfy the desired properties.
Thus, \emph{certifying} randomized constructions is often considered as a functional proxy when explicit constructions are hard to come by, as is the case for Euclidean sections.


This motivates the algorithmic task of certifying that a random subspace is well-spread.
With $\exp(d)$ time, this problem is trivial as we can brute-force search over an $\eps$-net in $\R^d$.
The best known non-trivial result is by Barak et.\ al.~\cite{BarakBHKSZ12}, who showed that given a random matrix $A \in \R^{n\times d}$ with i.i.d.\ Gaussian entries and $d \lesssim \sqrt{n}$, there is a polynomial-time algorithm certifying that $\colspan(A)$ has $O(1)$ distortion (with high probability over $A$)\footnote{This was implicitly proved in \cite{BarakBHKSZ12} and the result holds generally for matrices with sub-gaussian entries. An explicit proof was given in \cite{ChendO22}.}.
However, no efficient certification algorithm is known when $d \gg \sqrt{n}$
(note that the problem is harder for larger $d$).

More recently, ``evidence'' of computational hardness was established for the $d \gg \sqrt{n}$ regime in the form of lower bounds against low-degree polynomials~\cite{MaoW21,ChendO22}.
This suggests that there is an \emph{information-computation gap} for this problem with the computational threshold at $d \sim \sqrt{n}$, i.e., any polynomial-time algorithm for $d \gg \sqrt{n}$ would require significant breakthroughs (see e.g.\ \cite{Hopkins18,KuniskyWB19} for expositions of the low-degree hardness framework).

The lower bounds by \cite{MaoW21,ChendO22} leave open the possibility of \emph{subexponential}-time certification algorithms when $d$ is between $\sqrt{n}$ and $n$.
This sets the stage for our first result.

\begin{mtheorem}[Informal \Cref{thm:spread}] \label{thm:certification}
    Fix $\eps \in (0,1)$,
    and let $d, n \in \N$ such that $d = O(n^{\frac{1+\eps}{2}} / \log n)$.
    Let $A \sim \calN(0,1)^{n \times d}$.
    Then, there is a certification algorithm that runs in $2^{\wt{O}(n^{\eps})}$ time and, with probability $1- o(1)$ over $A$, certifies that $\colspan(A)$ has $O(1)$ distortion.
\end{mtheorem}

Notice that \Cref{thm:certification} establishes a smooth trade-off between the dimension $d$ and the runtime.
Specifically, when $\eps = 0$, we have $d \leq \wt{O}(\sqrt{n})$ and the runtime is polynomial, matching the best known algorithm by \cite{BarakBHKSZ12}.
When $\eps$ increases to $1$, the dimension $d$ increases to $n$ while the runtime increases to exponential.

We note that the phenomenon of information-computation gap is widespread across a variety of certification and inference problems.
For these problems, there is a parameter regime (often called the ``hard'' regime) in which it is conjectured that no polynomial-time algorithm exists.
In many cases, smooth trade-offs were established between runtime and the problem parameters, similar to that of \Cref{thm:certification}.
Examples of such trade-offs include runtime vs.\ the number of constraints in refuting random constraint satisfaction problems~\cite{RaghavendraRS17,GuruswamiKM22,HsiehKM23}, runtime vs.\ the approximation factor in polynomial optimization or tensor PCA~\cite{BhattiproluGGLT17,WeinAM19,HsiehKPT24}, runtime vs.\ the sparsity parameter in certifying the restricted isometry property (RIP) for random matrices~\cite{KoiranZ14,DingKWB21}, and runtime vs.\ the sparsity of the hidden vector in sparse PCA~\cite{dOrsiKNS20,DingKWB23}.

\paragraph{Finding planted sparse vector in a random subspace.}
For average-case optimization problems, techniques for certification can often be adapted to solving the \emph{search} problem for the related \emph{planted} model; this is called the \emph{proofs-to-algorithms} paradigm~\cite{FlemingKP19} in the literature.
For the problem of certifying distortion of a random subspace (in which there is no sparse vector), the corresponding planted problem is to find a \emph{sparse} vector planted in a random subspace.

We first formally define the planted sparse vector problem:

\begin{model}[Planted sparse vector problem]
    \label{model:planted-dist}
    Fix an unknown unit vector $v\in \R^n$, and let $d \leq n \in \N$.
    Let $\wt{A}$ be a random $n \times d$ matrix sampled as follows: (1) let $A$ be the random matrix such that the first column is $v$ and the other $d-1$ columns are i.i.d.\ $\calN(0, \frac{1}{n}\Id_n)$ vectors;
    (2) let $R \in \R^{d \times d}$ be an arbitrary unknown rotation matrix;
    (3) set $\wt{A} = A R$.

    The task is that given $\wt{A}$, output a unit vector $\wh{v} \in \R^n$ such that $\angles{\wh{v}, v}^2 \geq 1 - o(1)$.
\end{model}

For concreteness and comparison to prior works, we will focus on the special case where $v$ has noisy Bernoulli-Rademacher entries:

\begin{definition}[Noisy Bernoulli-Rademacher distribution~\cite{ChendO22}] \label{def:nBR}
    Given parameter $\rho \in (0,1)$ and $\sigma \in [0, 1/\sqrt{1-\rho})$, we define $\nBR(\rho,\sigma)$ to be the random variable such that
    \begin{equation*}
        x = \begin{cases}
            \calN(0, \sigma^2/n) & \text{with probability $1- \rho$,} \\
            +\frac{1}{\sqrt{\rho' n}} & \text{with probability $\rho/2$,} \\
            -\frac{1}{\sqrt{\rho' n}} & \text{with probability $\rho/2$,}
        \end{cases}
    \end{equation*}
    where $\rho' \coloneqq \frac{\rho}{1 - (1-\rho)\sigma^2}$.
\end{definition}

The parameter $\rho'$ is set such that $\E_{x\sim \nBR(\rho,\sigma)}[x^2] = 1/n$, so with high probability a vector $v \sim \nBR(\rho,\sigma)^n$ will have $\|v\|_2 \in (1 \pm o(1))$.
For concreteness, one can think of $\rho,\sigma$ as $\frac{1}{\polylog(n)}$.
This distribution is used in the hardness result of \cite{ChendO22}, and is the noisy version of the ``noiseless'' Bernoulli-Rademacher distribution considered in \cite{MaoW21,DiakonikolasK22,ZadikSWB22}\footnote{Surprisingly, \cite{DiakonikolasK22,ZadikSWB22} showed that one can recover $v$ even when $n = d+1$ if $v$ is a \emph{noiseless} Bernoulli-Rademacher vector, but their algorithms break down if one adds an inverse-polynomial amount of noise to $v$.}
and used as the ``spike'' in related sparse recovery problems like sparse PCA~\cite{JohnstoneL09,AminiW08,dOrsiKNS20,DingKWB23}.

The problem of finding the planted sparse vector was introduced by Spielman, Wang, and Wright~\cite{SpielmanWW12} in the context of dictionary learning.
This problem has since received a lot of attention due to various applications in learning theory and optimization (see e.g.~\cite{DemanetH13}).
Barak et.~al.~\cite{BarakKS14} was the first example of the proofs-to-algorithms paradigm for this problem.
Specifically, their algorithm uses the certification algorithm from \cite{BarakBHKSZ12} as a main ingredient, and thus it works when $d \ll \sqrt{n}$ (same as the polynomial-time regime for certification).
Various other algorithms have been proposed~\cite{DemanetH13,QuSW14,MaoW21,DiakonikolasK22,ZadikSWB22}, and the current best known algorithm is by Mao and Wein~\cite{MaoW21} which succeeds when $\rho d \ll \sqrt{n}$ (rather than $d \ll \sqrt{n}$ required in \cite{BarakKS14,HopkinsSSS16}).

On the other hand, for $d \gg \sqrt{n}$ (the hard regime for certification), the lower bounds by \cite{MaoW21,ChendO22} mentioned earlier also apply to the planted problem, as their hardness results are for the (easier) \emph{detection} problem.
Specifically, Chen and d'Orsi~\cite{ChendO22} showed that when $\rho d \geq \wt{\Omega}(\sqrt{n})$, all $\polylog(n)$-degree polynomials fail to distinguish between $A \sim \calN(0,1/n)^{n\times d}$ and $\wt{A}$ sampled from \Cref{model:planted-dist} with planted vector $v \sim \nBR(\rho,\sigma)^n$ for $\rho \gg 1/n$ and $\sigma \leq 1/\polylog(n)$.

Thus, it is natural to consider adapting our techniques for subexponential-time certification to the planted problem in the hard regime.
This is our second result.

\begin{mtheorem}[Informal \Cref{thm:planted-sparse-vector}]
\label{thm:planted-main}
    Let $t \leq d \leq n \in \N$ and $\rho \leq \frac{1}{\polylog(n)}$ such that $t \geq \frac{\rho d^2}{n} \polylog(n)$.
    There is a randomized algorithm with running time $2^{\wt{O}(t)}$ with the following guarantee:
    Given $\wt{A} \in \R^{n\times d}$ drawn from \Cref{model:planted-dist} with planted vector $v \sim \nBR(\rho,\frac{1}{\polylog(n)})^n$, with high probability, the algorithm outputs a unit vector $\wh{v} \in \colspan(\wt{A})$ such that $\angles{v, \wh{v}}^2 \geq 1-o(1)$.
\end{mtheorem}

For intuition, consider parameters $\rho = 1/{\polylog(n)}$ and $d = \wt{O}(n^{\frac{1+\eps}{2}})$ for $\eps \in (0,1)$.
Then, the algorithm recovers $v$ in $2^{\wt{O}(n^\eps)}$ time, the same dimension vs.\ runtime trade-off as established in \Cref{thm:certification}.
Our algorithm in fact works for more general sparse vectors $v$;
see \Cref{thm:planted-sparse-vector} and \Cref{rem:assumption-on-v} for details and discussions on the assumptions of $v$ that we require.


\begin{remark}
    One can consider the harder planted problem where you are given an arbitrary orthogonal basis of the subspace instead of a ``Gaussian basis'' like \Cref{model:planted-dist}.
    In fact, both \cite{HopkinsSSS16} and \cite{MaoW21} started with \Cref{model:planted-dist}, then with some extra work, they showed that the algorithms are robust to exchanging the Gaussian basis for an arbitrary orthogonal basis.
    We leave this as a future work, though we believe similar techniques (like matrix perturbation analysis) used in \cite{HopkinsSSS16,MaoW21} may apply to our case as well.
    On the other hand, all hardness results~\cite{MaoW21,ChendO22} are proved for \Cref{model:planted-dist}.
\end{remark}




\section{Technical Overview}
\label{sec:technical-overview}

\paragraph{Notation.}
For an integer $N$, we will use $[N]$ to denote the set $\{1,2,\dots, N\}$.
For a vector $x \in \R^n$, we use $\|x\|_p$ to denote its $\ell_p$ norm, and for any $S\subseteq [n]$, we write $x_S \in \R^{|S|}$ to denote the vector restricted to coordinates in $S$.
For a matrix $M \in \R^{n\times d}$, we use $\|M\|_2$ to denote its operator (spectral) norm: $\sup_{x\neq 0} \frac{\|Mx\|_2}{\|x\|_2}$.
Moreover, for any $S \subseteq [n]$, we write $M_S \in \R^{|S|\times d}$ to denote the submatrix of $M$ obtained by selecting rows according to $S$.

\paragraph{Organization.}
\Cref{sec:BBH} describes the certification algorithm for $d \leq O(\sqrt{n})$ by \cite{BarakBHKSZ12} using the $2$-to-$4$ norm of $A$.
\Cref{sec:beyond-sqrt-n} explains the barrier of going beyond $\sqrt{n}$ as well as our strategy to bypass it by removing outlier entries of $Ax$.
This motivates a key ingredient in our analysis, which is the elementary symmetric polynomial $P_t(x)$ explained in \Cref{sec:symmetric-polynomials}.
In \Cref{sec:trace-method}, we make a brief detour to explain the trace moment method and the analysis of \cite{Tao12}.
In \Cref{sec:upper-bound-Pt}, we upper bound $P_t(x)$ by bounding the spectral norm of a related random matrix (which we can certify) using the trace method, and explain the crucial idea that symmetrization lowers the spectral norm.
Finally, in \Cref{sec:planted-overview}, we give an overview of how we adapt our techniques to the planted problem.

\subsection{2-to-4 norm as a proxy for sparsity: certification algorithm of \texorpdfstring{\cite{BarakBHKSZ12}}{[BBH+12]}}
\label{sec:BBH}

Consider a matrix $A \in \R^{n \times d}$ with i.i.d.\ $\calN(0,1)$ entries (where $d \ll n$), and let $a_1, a_2,\dots, a_n \in \R^d$ be its rows.
To certify that $\colspan(A)$ has small distortion (as defined in \Cref{eq:distortion}), one needs to certify that $Ax$ is ``sparse'' --- $\|Ax\|_2 / \|Ax\|_1 \leq O(1/\sqrt{n})$ --- for all $x\in \R^d$.
It is well known that $A$ has singular values between $\sqrt{n} (1 \pm o(1))$ with high probability (see \Cref{fact:gaussian-matrix}), so $\|Ax\|_2 \approx \sqrt{n}$ for all unit vectors $x$.
On the other hand, as intuition, consider a random unit vector $x$.
Each $|\angles{a_i,x}|$ is roughly $\Omega(1)$, hence
$\|Ax\|_1 = \sum_{i=1}^n |\angles{a_i, x}| \geq \Omega(n)$.
So, for a random $x$ we have $\|Ax\|_2 / \|Ax\|_1 \leq O(1/\sqrt{n})$, as desired.

The difficult part is to certify a lower bound on $\|Ax\|_1$ for \emph{all} unit vectors $x$.
Thus, in many applications, it is more tractable to consider an alternative proxy for sparsity --- the \emph{$2$-to-$4$ norm}.
Intuitively, vectors with small $4$-norm compared to the $2$-norm are considered well-spread.
The quantity $\max_{x\neq 0} \|Ax\|_4/ \|x\|_2$, denoted $\|A\|_{2\to 4}$, is called the \emph{$2$-to-$4$ norm} of $A$, and upper bounds on $\|A\|_{2\to 4}$ are called \emph{hypercontractivity} inequalities.

Barak et.~al.~\cite[Theorem 7.1]{BarakBHKSZ12} showed that with high probability,
\begin{equation*}
    \|A\|_{2\to 4}^4 \leq n \cdot \parens*{3 + O(1) \cdot \max\parens*{\frac{d}{\sqrt{n}}, \frac{d^2}{n}}} \mper
\end{equation*}
Moreover, this can be efficiently certified via a natural SDP relaxation for maximizing $\|Ax\|_4^4 = \sum_{i=1}^n \angles{a_i,x}^4$, a degree-4 polynomial.
Therefore, when $d \leq O(\sqrt{n})$, we have that $\sum_{i=1}^n \angles{a_i, x}^4 \leq O(n)$ for all unit vectors $x$.
Note that this bound matches the case when most $|\angles{a_i,x}|$ are roughly $\Theta(1)$, which is the case for a random $x$.
Combined with standard tools (\Cref{lem:spread-distortion} and \cite[Proposition 3.4]{ChendO22}) and the fact that $\|Ax\|_2 \approx \sqrt{n}$, this implies that $\colspan(A)$ has $O(1)$ distortion.

We now give a brief overview of the proof by \cite{BarakBHKSZ12}, which is our starting point.
The main idea is to write $\|Ax\|_4^4$ as follows,
\begin{equation*}
    \|Ax\|_4^4 = \sum_{i=1}^n \angles{a_i,x}^4
    = (x^{\otimes 2})^\top \parens*{\sum_{i=1}^n a_i^{\otimes 4}} (x^{\otimes 2}) \mcom
\end{equation*}
where we view $x^{\otimes 2}$ as a $d^2$-dimensional vector and $\sum_{i=1}^n a_i^{\otimes 4}$ as a $d^2 \times d^2$ random matrix.
The main idea of \cite{BarakBHKSZ12} is that by standard matrix concentration inequalities (like the Hanson-Wright inequality~\cite{Ver20} and results from \cite{AdamczakLPT11}), $\sum_{i=1}^n a_i^{\otimes 4}$ concentrates around its mean $\E\sum_{i=1}^n a_i^{\otimes 4} = n \Sigma$ where $\Sigma = \E[a_i^{\otimes 4}]$, again viewed as a $d^2 \times d^2$ matrix.

The naive approach is to bound the operator norm $\|\Sigma\|_2$, then we have $\|Ax\|_4^4 \leq n \cdot \|\Sigma\|_2$.
Unfortunately, there is a rank-$1$ component $ww^\top$ in $\Sigma$ where $w_{ij} = \delta_{ij}$, and $\|ww^\top\|_2 = n$.
The crucial observation is to ``shift'' the entries of $\Sigma$ using the symmetry of $x^{\otimes 2}$ to ``break'' the rank-$1$ component.
Then, the shifted matrix $\shift(\Sigma)$ satisfies $(x^{\otimes 2})^\top \Sigma (x^{\otimes 2}) = (x^{\otimes 2})^\top \shift(\Sigma) (x^{\otimes 2}) \leq \|\shift(\Sigma)\|_2 \leq O(1)$.
This implies that $\|Ax\|_4^4 \leq O(n)$, as desired.
As we will explain later, this symmetrization technique to decrease the matrix norm is central to our analysis and is also the key idea behind the results of \cite{RaghavendraRS17,BGL-random17,BhattiproluGGLT17}.



\subsection{Going beyond \texorpdfstring{$\sqrt{n}$}{sqrt(n)}}
\label{sec:beyond-sqrt-n}

When $d \gg \sqrt{n}$, we face the immediate problem that $\|A\|_{2\to4}^4 \leq n$ is no longer true.
One can easily see this by considering $x = a_1 / \|a_1\|_2$.
In this case, we have $\angles{a_1, x} = \|a_1\|_2 \approx \sqrt{d}$ since $a_1$ is a $d$-dimensional Gaussian vector, and $|\angles{a_i, x}| \approx \Theta(1)$ for $i \neq 1$.
Then, $\|Ax\|_4^4 = \sum_{i=1}^n \angles{a_i,x}^4 \geq \Omega(d^2 + n) = \Omega(d^2) \gg n$.

More generally, for any $S \subseteq [n]$ of size $t \ll d$, the submatrix $A_S \in \R^{t\times d}$ has maximum singular value $\approx \sqrt{d}$, i.e., there is an $x$ such that $\|A_S x\|_2^2 \approx d$.
Suppose the vector $A_S x$ is roughly equally distributed, then we have $\|A_S x\|_4^4 \geq (\frac{d}{t})^2 t = \frac{d^2}{t}$, which means that $\|Ax\|_4^4 \geq \frac{d^2}{t}$.

\parhead{Removing the top entries of $Ax$.}
These counter-examples give an important insight: $\|Ax\|_4^4$ seems to be dominated by just a few very large entries in $Ax$.
In the setting of \Cref{thm:certification}, we have $d = n^{\frac{1+\eps}{2}}$, and in the example above, $\|Ax\|_4^4 \geq \frac{d^2}{t} \gg n$ when $t \ll n^{\eps}$.
One might guess that after removing the top $n^{\eps}$ entries of $Ax$, the resulting vector has small $4$-norm.

We show that this is indeed the case:
\begin{lemma}[Informal \Cref{lem:2-to-4-excluding-top}]
\label{lem:excluding-top-informal}
    For any unit vector $x\in \R^d$, let $T_x \subseteq [n]$, $|T_x| = n - n^{\eps}$ be the set of indices excluding the top $n^{\eps}$ entries of $Ax$.
    Then, $\|A_{T_x} x\|_4^4 \leq O(n)$ with high probability over $A$.
    Moreover, there is an algorithm that certifies this in $2^{\wt{O}(n^{\eps})}$ time.
\end{lemma}

Given \Cref{lem:excluding-top-informal}, \Cref{thm:certification} follows almost immediately.
The proof is given in \Cref{sec:certification}.

The first part that $\|A_{T_x}x\|_4^4 \leq O(n)$ can in fact be proved via a probabilistic argument.
The challenge is to certify it.
At a high-level, we need a ``proxy'' for $\|Ax\|_4^4$ that is robust to $t \coloneqq n^{\eps}$ large entries of $Ax$ (obviously without knowing where the large entries are because $T_x$ depends on $x$).
In light of this, we turn our attention to the \emph{elementary symmetric polynomials}.

\subsection{Elementary symmetric polynomials}
\label{sec:symmetric-polynomials}

We define
\begin{equation*}
    Q_t(z) \coloneqq t! \sum_{S \subseteq[n]: |S|=t} \prod_{i\in S} z_i^4 \mcom
\end{equation*}
and for $A$ with rows $a_1,a_2,\dots,a_n\in \R^d$,
\begin{equation*}
    P_t(x) \coloneqq Q_t(Ax) = t! \sum_{S\subseteq [n]: |S|=t} \prod_{i\in S} \angles{a_i, x}^4 \mcom
\end{equation*}
where we omit the dependence on $a_1,\dots,a_n$ for simplicity.

The scaling $t!$ here is because $Q_t(z)$ is exactly the ``multilinear'' terms when one expands out $\|z\|_4^{4t} = (\sum_{i=1}^n z_i^4)^t$, and thus our target upper bound is $O(n)^t$.
Moreover, $P_t(x)$ can be computed in $n^{O(t)}$ time.

The following is our key lemma.

\begin{lemma}[Informal \Cref{lem:sum-of-distinct-products}]
\label{lem:Pt-bound-informal}
    With high probability over $A$, there is an algorithm that runs in $n^{O(t)}$ time and certifies that for all unit vectors $x\in \R^d$,
    \begin{equation*}
        P_t(x) \leq O\parens*{n + \frac{d^2}{t} \log^2 n}^t \mper
    \end{equation*}
\end{lemma}

When $d \leq \wt{O}(n^{\frac{1+\eps}{2}})$ and $t = n^{\eps}$, we get $P_t(x) \leq O(n)^t$.
This is what we expect for a typical $x$, where most $|\angles{a_i,x}|$ are $O(1)$ and $P_t(x) \approx n^t \cdot O(1)^{4t} = O(n)^t$.

A careful reader may notice that an upper bound on $Q_t(z)$ does not imply that the vector $z$ is dense.
For example, $Q_t(z) = 0$ whenever $z$ has less than $t$ nonzero entries.
Thus, \Cref{lem:Pt-bound-informal} does not immediately imply that $Ax$ is dense for all $x$.
However, notice that we are allowed $n^{O(t)}$ time, so we can exhaustively search over all subsets $S \subseteq [n]$ of size $t$ and ``check'' that no $Ax$ is concentrated on $t$ coordinates.
This comes down to computing the maximum singular value of $A_S$ for each $S$.

The overall analysis (proof of \Cref{lem:excluding-top-informal} from \Cref{lem:Pt-bound-informal}) is as follows.
For any unit vector $x\in \R^d$, since $\ol{T}_x$ is the set with the $t$ largest $\angles{a_i,x}^4$, and $\|A_{\ol{T}_x} x\|_2^2 \leq \sigma_{\max}(A_{\ol{T}_x})^2 \leq d(1+o(1))$ verified by exhaustive search, it follows that $\angles{a_i,x}^2 \leq \frac{d}{t}(1+o(1)) \coloneqq \tau$ for all $i\in T_x$.
Then, we consider $\|A_{T_x} x\|_4^{4t} = (\sum_{i\in T_x} \angles{a_i,x}^4)^t$.
For the multilinear term in the expansion of $\|A_{T_x} x\|_4^{4t}$, we have
\begin{equation*}
    t! \sum_{S \subseteq T_x: |S|=t} \prod_{i\in S} \angles{a_i,x}^4
    \leq t! \sum_{S \subseteq [n]: |S|=t} \prod_{i\in S} \angles{a_i,x}^4
    = P_t(x) \leq O(n)^t \mper
\end{equation*}
Note that here we changed the summation of $S \subseteq T_x$ to $S \subseteq [n]$ (since $T_x \subseteq [n]$ and all terms are non-negative).
Crucially, we get an upper bound that does not depend on $T_x$, thus allowing us to use \Cref{lem:Pt-bound-informal}.

For the rest of the terms in $(\sum_{i\in T_x} \angles{a_i,x}^4)^t$, we use combinations of \Cref{lem:Pt-bound-informal} and $\angles{a_i,x}^2\leq \tau$ to bound, which ultimately proves \Cref{lem:excluding-top-informal}.

\subsection{Overview of the trace moment method}
\label{sec:trace-method}

In this section, we make a detour and give an overview of the trace moment method.
Readers familiar with the analysis of random matrix norm bounds by Tao~\cite{Tao12} may skip directly to \Cref{sec:upper-bound-Pt}.
The starting point of the trace moment method is the inequality
\begin{equation*}
    \frac{1}{N} \cdot \tr(M^{\ell}) \leq \|M\|_2^\ell \leq \tr(M^{\ell}) \mcom
\end{equation*}
which holds for any $N \times N$ symmetric matrix $M$ and any even integer $\ell$.
Expanding $\tr(M^\ell)$ and taking expectation (when $M$ is a random matrix), we get
\begin{equation*}
    \E\bracks*{\|M\|_2^\ell} \leq \E\bracks*{\tr(M^\ell)}
    = \E \sum_{i_1,i_2,\dots,i_\ell \in [N]} M_{i_1,i_2} M_{i_2,i_3} \cdots M_{i_{\ell},i_1} \mper
\end{equation*}
Notice that the summation is over closed walks on $[N]$: $i_1 \to i_2 \to \cdots \to i_{\ell} \to i_1$.
We can view these as closed walks on the complete graph with $N$ vertices.

We now briefly explain Tao's analysis~\cite{Tao12} on bounding $\E[\tr(M^{\ell})]$ for the (symmetric) random sign matrix, where each entry is uniformly random $\pm 1$.
The most important observation is that if a closed walk uses any edge (in the $N$-vertex complete graph) an \emph{odd} number of times, then the expectation is zero.
Thus, we only need to consider closed walks that use each edge an even number of times, each of which contributes $1$ in the trace.

\parhead{Encoding closed walks.}
The main idea is to \emph{combinatorially} count all such ``even'' closed walks by providing an \emph{encoding}.
Following the terminology of \cite{Tao12}, we refer to a step $i_k \to i_{k+1}$ in the walk as a \emph{leg}.
We call a leg ``fresh'' if it traverses an edge that we haven't seen before, and ``return'' if it traverses a previously used edge.
Since each edge must be traversed at least twice, we can have at most $\ell/2$ fresh legs.
We will view a fresh leg as creating an ``active'' edge, while a return leg closes an active edge; all edges must be closed in the end.
Now, we can encode a closed walk as follows: (1) pick a starting vertex in $[N]$, (2) label each leg as ``fresh'' or ``return'' such that there are at most $\ell/2$ fresh legs, (3) for a fresh leg we have $N$ choices to choose the destination, and for now, assume that each return leg has a unique choice so that we don't need to specify a destination.
Then, we upper bound $\E[\tr(M^\ell)]$ by upper bounding the number of encodings:
\begin{equation*}
    \E\bracks*{\tr(M^\ell)} \leq N \cdot 2^{\ell} \cdot N^{\ell/2} = N \cdot (2\sqrt{N})^{\ell} \mper
\end{equation*}
Then, setting $\ell \gg \log N$, by Markov's inequality, with probability $1-\frac{1}{\poly(N)}$ we have
\begin{equation*}
    \|M\|_2 \leq \tr(M^\ell)^{1/\ell} \leq N^{O(1/\ell)} \cdot 2\sqrt{N} \leq (2 + o(1)) \sqrt{N} \mper
\end{equation*}
This is the desired upper bound with the correct constant factor $2$.

\parhead{Unforced return legs.}
There is an important assumption that we made in the above analysis: each return leg has a unique choice.
This is obviously not true for all walks, but we will prove that this is true for \emph{most} walks.
In \cite{Tao12}, a return leg with one unique choice is called \emph{forced}, and ones with multiple choices are called \emph{unforced}.
We need to prove that the closed walks with any unforced return legs are negligible.
The key observation is that if a return leg is unforced at vertex $v$, i.e., there are multiple active edges incident to $v$, it must be the case that there were previous fresh legs that went back to $v$.
Such legs are called \emph{non-innovative}, as they create new active edges but not new vertices, and as a result, each non-innovative leg requires only a factor $\ell$ --- the length of the walk --- to specify its destination (as opposed to $N$ for fresh legs).

The main idea is to \emph{charge} the extra information required for unforced return legs to the non-innovative legs.
We need to prove (1) an upper bound on the unforced return legs in terms of the number of non-innovative legs, and (2) the cost of a non-innovative leg (with extra information) is still $\ll N$, negligible compared to a fresh leg.
This is a key technical challenge in most prior works on (sharp) norm bounds for more complicated random matrices~\cite{Tao12,JonesPRTX22,HsiehKPX23}.



\subsection{Upper bound on \texorpdfstring{$P_t(x)$}{Pt(x)}}
\label{sec:upper-bound-Pt}
We now give an overview of the proof of \Cref{lem:Pt-bound-informal}, which is the most technical part.

We start by writing $\angles{a_i, x}^4 = \angles{a_i^{\otimes 4}, x^{\otimes 4}} = \angles{\shift(a_i^{\otimes 4}), x^{\otimes 4}}$, where $\shift$ is formally defined in \Cref{def:shift-entries} and is also done in \cite{BarakBHKSZ12} to remove the large rank-1 component in the matrix (see \Cref{rem:lower-bound-without-shifting} for a lower bound without this).
Then,
\begin{equation*}
    P_t(x) = \angles*{t! \sum_{S\subseteq[n]: |S|=t} \bigotimes_{i\in S} \shift(a_i^{\otimes 4}), x^{\otimes 4t}}
    = (x^{\otimes 2t})^\top M (x^{\otimes 2t}) \mcom
\end{equation*}
where we view $x^{\otimes 2t}$ as a vector of dimension $d^{2t}$, and
\begin{equation*}
    M \coloneqq t! \sum_{S\subseteq[n]:|S|=t} \bigotimes_{i\in S} \shift(a_i^{\otimes 4}) \in \R^{d^{2t} \times d^{2t}} \mper
\end{equation*}
Here, we view $a_i^{\otimes 4}$ as a $d^{2} \times d^{2}$ matrix.

The natural approach to bound $P_t(x)$ is to bound the spectral norm of $M$, since $P_t(x) \leq \|M\|_2 \cdot \|x^{\otimes 2t}\|_2^2 = \|M\|_2$ for any unit vector $x$.
Unfortunately, it can be shown that $\|M\|_2 \geq \Omega(d)^{2t} \gg n^t$ (see \Cref{rem:trace-lower-bound}).
In fact, using $\|M\|_2$ to bound $P_t(x)$ does not improve on \cite{BarakBHKSZ12} at all.

\parhead{Symmetrization lowers spectral norm.}
To remedy this, observe that the flattened vector $x^{\otimes 2t}$ in the quadratic form is highly symmetric.
Specifically, let $\Pi \in \R^{d^{2t} \times d^{2t}}$ be \emph{any} permutation matrix that maps an index $J = (j_1,j_2,\dots,j_{2t})\in [d]^{2t}$ to $(j_{\pi(1)}, j_{\pi(2)}, \dots, j_{\pi(2t)})$ for some permutation $\pi$ over $2t$ elements.
Then, we have that
\begin{equation*}
    \Pi x^{\otimes 2t} = x^{\otimes 2t}\mper
\end{equation*}
We will use $\Permutation{2t}$ to denote the collection of all such matrices $\Pi$.
It follows that for any $\Pi, \Pi'\in \Permutation{2t}$,
\begin{equation*}
    (x^{\otimes 2t})^\top M (x^{\otimes 2t}) = (x^{\otimes 2t})^\top (\Pi M \Pi') (x^{\otimes 2t}) \mper
\end{equation*}
In particular,
\begin{equation*}
    (x^{\otimes 2t})^\top M (x^{\otimes 2t}) = (x^{\otimes 2t})^\top \E_{\Pi,\Pi'\sim \Permutation{2t}}[\Pi M \Pi'] (x^{\otimes 2t})
    \leq \norm{\wt{M}}_2 \mcom
\end{equation*}
where $\wt{M} = \E_{\Pi,\Pi'\sim \Permutation{2t}}[\Pi M \Pi']$.

As $\wt{M}$ is the symmetrized version of $M$, one may expect that the spectral norm decreases.
Indeed, this is also the main idea behind the results of \cite{RaghavendraRS17,BGL-random17,BhattiproluGGLT17}.
Our key technical result is that the symmetrization lowers the spectral norm by a factor of roughly $t^t$, and we prove it using the trace moment method.

\begin{lemma}[Informal \Cref{lem:trace-bound}]
\label{lem:trace-informal}
    Let $\ell = \log n$. Then,
    \begin{equation*}
        \E \tr(\wt{M}^\ell) \leq d^{2t} \cdot O\parens*{n + \frac{d^2}{t}\log^2 n}^{t\ell} \mper
    \end{equation*}
\end{lemma}
Here, we need $\ell = \log n$ (as opposed to $\log(\dim(\wt{M})) = 2t\log d$) since we only need $d^{2t/\ell} \leq O(1)^t$, and saving this $t$ factor turns out to be important in our analysis.
From the discussion above, \Cref{lem:Pt-bound-informal} follows immediately since $P_t(x) = (x^{\otimes 2t})^\top \wt{M} (x^{\otimes 2t}) \leq \|\wt{M}\|_2 \leq \tr(\wt{M}^{\ell})^{1/\ell}$.
We also remark that \Cref{lem:trace-informal} is tight; see \Cref{rem:trace-lower-bound}.

\parhead{Trace method for $\wt{M}$.}
First, note that $\E \tr(\wt{M}^\ell) = \E \tr \E_{\Pi_1,\Pi_1',\dots,\Pi_\ell,\Pi_\ell'}[\Pi_1 M \Pi_1' \cdots \Pi_{\ell} M \Pi_{\ell}']$.
The trace is the sum of length-$\ell$ closed walks $I_1 \to I_2 \to \cdots \to I_\ell \to I_1$ on the indices (in $[d]^{2t}$).
For each sequence $\bpi = (\Pi_1,\Pi_1',\dots, \Pi_\ell, \Pi_\ell')$, this corresponds to a \emph{labeling} of a specific \emph{structure} defined by $\bpi$, viewed as a graph containing circle and square vertices (see \Cref{def:structure,def:valid-labeling}, and \Cref{fig:structure-examples} for examples).
The circle and square vertices receive labels in $[d]$ and $[n]$ respectively.
We require that each labeled edge (an element in $[n] \times [d]$) appears an even number of times, otherwise the expectation is zero.

We next define an encoding of the labeling such that 
(1) every valid labeling for a given structure can be encoded, and 
(2) a labeling is uniquely determined given an encoding and a structure (if decoding succeeds).
Similar to the analysis in \Cref{sec:trace-method}, we mark each leg as ``Fresh'', ``Return'', ``Non-innovative'' or ``High-multiplicity''.
We also introduce an extra ``Paired'' type since in our case a pair of two circle vertices lead to the same square vertex.
In addition, we need to handle potential unforced return legs by charging extra information to the non-innovative and high-multiplicity legs.
These are all carefully done in \Cref{sec:encoding}.

\parhead{Decoding success probability.}
We now formalize the idea that symmetrization reduces the spectral norm as follows,
\begin{equation*}
    \E \tr(\wt{M}^\ell) \leq \E_{\bpi} \sum_{\sigma:\text{ labelings}} \1(\sigma \text{ valid for } \bpi)
    \leq \sum_{\xi:\text{ encoding}} \Pr_{\bpi}[\Decode(\xi,\bpi) \text{ succeeds}] \mper
\end{equation*}
We upper bound the probability of the decoder succeeding (for any encoding $\xi$) in \Cref{lem:decoding-success-probability}.
At a high level, for any structure (\Cref{fig:structure-examples}), the circle vertices are paired and connected to a square vertex.
During decoding, each square vertex receives two labels from the two legs (coming from previous circles), and the decoding fails if there is a conflict.
Now, when we randomize the structure $\bpi$, which we can view as rewiring the edges in the structure and pairing up circle vertices within each block, the decoding succeeds only if all circle vertices are paired up in the correct way.
For each block, there are $2t$ circle vertices, and there are $(2t-1)!! \geq \Omega(t)^t$ number of ways to group them into pairs.
Thus, for all $\ell$ blocks, we get a factor $t^{t\ell}$ improvement.
We defer the details to \Cref{sec:decoding}, and the proof of \Cref{lem:trace-informal} is completed in \Cref{sec:proof-of-trace-bound}.

\subsection{Finding planted sparse vector}
\label{sec:planted-overview}

It turns out that our upper bound on $P_t(x)$ can be used for the related planted problem.
More specifically, the proof of \Cref{lem:Pt-bound-informal} shows a stronger statement:
the upper bound on $P_t(x)$ exhibits a degree-$4t$ Sum-of-Squares (SoS) proof, meaning that any pseudo-distribution $\mu$ satisfies that $\pE_{\mu}[P_t(x)]$ is small (see \Cref{sec:sos} for background on SoS proofs and pseudo-distributions).

On the other hand, suppose $\wt{A} = AR$ is drawn from \Cref{model:planted-dist} with a planted sparse vector $v$ such that $Q_t(v)$ is large, and let $\wt{P}_t(x) \coloneqq Q_t(\wt{A}x)$.
Then, we must have $\max_{\mu} \pE_{\mu}[\wt{P}_t(x)] \geq Q_t(v)$, where the maximum is over all pseudo-distributions satisfying the unit sphere constraint.
This is because taking $x = r_1 \coloneqq R^\top e_1$ (the first row of $R$) gives $\wt{P}_t(r_1) = Q_t(v)$.
Moreover, the pseudo-distribution that maximizes $\wt{P}_t(x)$ can be computed in $n^{O(t)}$ time using the SoS algorithm.

Thus, the main intuition is that this pseudo-distribution $\mu$ must have significant support on vectors close to $r_1$, otherwise $\pE_{\mu}[\wt{P}(x)]$ cannot be large.
We formalize this in \Cref{lem:large-y1} and prove that $\pE_{\mu}[\angles{r_1,x}^{2t}] \geq e^{-o(t)} = (1-o(1))^t$.

Now, we use a rounding algorithm of \cite{BarakKS15} (\Cref{lem:rounding-powers}) to obtain a list of $2^{\wt{O}(t)}$ unit vectors with the guarantee that one of them, say $\wh{x}$, is close to $r_1$, which means that $\wh{v} = \wt{A} \wh{x}$ is close to $v$.
The next challenge is to identify a ``good'' $\wh{x}$ among this list.
The natural idea is to take the vector which is the ``most compressed'', i.e., the one whose top $\rho n$ entries have the largest norm.
\Cref{lem:selecting-the-vector} proves that this indeed works, which completes the proof.

\section{Certifying Spread of a Random Subspace}
\label{sec:certification}

Recall from \Cref{eq:distortion} that the distortion of a subspace $X$ is defined as $\Delta(X) = \sup_{x\in X: x\neq 0} \frac{\sqrt{n} \|x\|_2}{\|x\|_1}$.
In the following, we state an equivalent notion of spreadness.

\begin{definition}[Spreadness property of a subspace] \label{def:spread-subspace}
    A subspace $X \subseteq \R^n$ is $(t,\eps)$-spread if for every $x\in X$ and every $S\subseteq [n]$ with $|S| \leq t$, we have
    \begin{equation*}
        \norm*{x_{\ol{S}}}_2 \geq \eps \cdot \norm{x}_2 \mper
    \end{equation*}
\end{definition}

In other words, $X$ is $(t,\eps)$-spread if any vector $x\in X$ still has large ($\geq \eps$ fraction) norm after removing any $t$ coordinates.
The relationship between the distortion and spread of a subspace was proved in \cite{GuruswamiLR10}:

\begin{lemma}[Lemma 2.11 of \cite{GuruswamiLR10}]  \label{lem:spread-distortion}
    Let $X\subseteq \R^n$ be a subspace.
    \begin{enumerate}[(1)]
        \item If $X$ is $(t,\eps)$-spread, then $\Delta(X) \leq \sqrt{\frac{n}{t}} \cdot \eps^{-2}$.
        \item Conversely, $X$ is $\parens*{\frac{n}{2\Delta(X)^2}, \frac{1}{4\Delta(X)}}$-spread.
    \end{enumerate}
\end{lemma}

In light of \Cref{lem:spread-distortion}, to prove that $\Delta(X) \leq O(1)$, it suffices to prove that $X$ is $(\Omega(n),\Omega(1))$-spread.

We now state our main result for certification.

\begin{theorem}[Formal version of \Cref{thm:certification}] \label{thm:spread}
    Let $d, n \in \N$ and $\eps \in (0,1)$ such that $d = n^{\frac{1+\eps}{2}} / \log n$ and $n^{\eps} \log^6 n \leq n$.
    Let $A \sim \calN(0,1)^{n \times d}$.
    Then, there is a certification algorithm that runs in $n^{O(n^{\eps})}$ time and, with probability $1- \frac{1}{\poly(n)}$ over $A$, certifies that $\colspan(A)$ is $(\alpha n, 1/2)$-spread, where $\alpha$ is a universal constant.
\end{theorem}

We remark that our algorithm works not just for Gaussian matrices but also matrices with more general random variables; see \Cref{lem:trace-bound} for the requirements.


\subsection{Removing the top entries of \texorpdfstring{$Ax$}{Ax}}

As discussed in \Cref{sec:beyond-sqrt-n}, the barrier of the certification algorithm of \cite{BarakBHKSZ12} at $d \gg \sqrt{n}$ is that if $x$ is highly correlated with some row $a_i$ in $A$ (i.e., $|\angles{a_i,x}|$ is large), then $\|Ax\|_4^4$ can be much larger than $n$.
Our intuition is that the number of such large entries in $Ax$ must be very small.
Thus, if we remove the top few entries of $Ax$, then the $4$-norm is at most $O(n)$, the desired bound.

The next lemma states that this is indeed the case and that we can certify this.

\begin{lemma} \label{lem:2-to-4-excluding-top}
    Let $d, n \in \N$ and $\eps \in (0,1)$ such that $d = n^{\frac{1+\eps}{2}} / \log n$ and $n^{\eps} \log^6 n \leq n$.
    Let $a_1,\dots,a_n \in \R^d$ be random vectors with i.i.d.\ $\calN(0,1)$ entries.
    Then, there is a certification algorithm that runs in $n^{O(n^{\eps})}$ time and, with probability $1- \frac{1}{\poly(n)}$ over $a_1,\dots,a_n$, certifies the following:
    \begin{enumerate}[(1)]
        \item For any $S \subseteq[n]$ of size $|S| = n^{\eps}$, $\sum_{i\in S} \angles{a_i,x}^2 \leq (1+o_n(1)) \cdot d$ for all unit vectors $x\in \R^d$.
        \label{item:submatrix-singular-value}

        \item For any unit vector $x \in \R^d$,
        \begin{equation*}
            \sum_{i\in T_x} \angles{a_i, x}^4 \leq O(n) \mper
        \end{equation*}
        where $T_x \subseteq [n]$, $|T_x| = n-n^{\eps}$ is the set of indices excluding the top $n^{\eps}$ with the largest $\angles{a_i, x}^2$.
        \label{item:4-norm}
    \end{enumerate}
\end{lemma}

Proving \ref{item:submatrix-singular-value} is straightforward by applying the following standard matrix concentration result and a union bound.

\begin{fact}[See e.g.~\cite{Ver20}] \label{fact:gaussian-matrix}
    Let $m \leq n$, and let $A \in \R^{n \times m}$ be a matrix with i.i.d.\ $\calN(0,1)$ entries.
    Then, for every $t \geq 0$, with probability at least $1 - 2e^{-t^2/2}$, we have $\sqrt{n} - \sqrt{m} - t \leq \sigma_{\min}(A) \leq \sigma_{\max}(A) \leq \sqrt{n} + \sqrt{m} + t$.
\end{fact}

For \ref{item:4-norm}, first note that \ref{item:submatrix-singular-value} implies that $\angles{a_i, x}^2 \leq \frac{d}{n^{\eps}} (1+o_n(1))$ for all $i\in T_x$.
Then, we consider $\parens*{\sum_{i\in T_x} \angles{a_i,x}^4}^t$.
For starters, let's focus on the term $t! \sum_{S \subseteq T_x: |S|=t} \prod_{i\in S} \angles{a_i, x}^4$, i.e., the terms where the indices are distinct.
The main observation is that we can upper bound this by $P_t(x) \coloneqq t! \sum_{S \subseteq [n]: |S|=t} \prod_{i\in S} \angles{a_i, x}^4$, where we replace $T_x$ with $[n]$.
Note that this quantity crucially does not depend on $T_x$.

The bulk of our proof is then to prove the following lemma:

\begin{lemma} \label{lem:sum-of-distinct-products}
    Let $t \leq d \leq n$ be integers such that $t \log^6 n \leq n + \frac{d^2}{t} \log^2 n$.
    Let $a_1,\dots,a_n \in \R^d$ be random vectors with i.i.d.\ $\calN(0,1)$ entries.
    Then, there is a certification algorithm that runs in $n^{O(t)}$ time and, with probability $1 - \frac{1}{\poly(n)}$ over $a_1,\dots,a_n$, certifies that for all unit vectors $x\in \R^d$,
    \begin{equation*}
        P_t(x) \coloneqq t!\sum_{S \subseteq [n]: |S|=t} \prod_{i\in S} \angles*{a_i, x}^4 \leq O\parens*{n + \frac{d^2}{t}\log^2 n}^t \mper
    \end{equation*}
\end{lemma}

We will prove \Cref{lem:sum-of-distinct-products} at the end of \Cref{sec:symmetrization-lowers-trace}.
We first use \Cref{lem:sum-of-distinct-products} to prove \Cref{lem:2-to-4-excluding-top}.

\begin{proof}[Proof of \Cref{lem:2-to-4-excluding-top} from \Cref{lem:sum-of-distinct-products}]
    Let $t \coloneqq n^{\eps}$, and let $A$ be the $n \times d$ matrix with $a_1,a_2,\dots,a_n$ as columns.
    The certification algorithm is as follows,
    \begin{enumerate}[(i)]
        \item For each $S \subseteq [n]$ with $|S| = t$, verify that $\sigma_{\max}(A_S) \leq (1+o_n(1)) \sqrt{d}$, where $A_S \in \R^{t \times d}$ is the submatrix of $A$ obtained by choosing rows according to $S$.
        \item For $s = 1,2,\dots, t$, use the algorithm in \Cref{lem:sum-of-distinct-products} to certify that
        \begin{equation*}
            P_s(x) \coloneqq s! \sum_{S\subseteq [n]: |S|=s} \prod_{i\in S} \angles{a_i,x}^4 \leq O\parens*{\frac{d^2}{s}\log^2 n}^s
            \quad \text{for all unit vectors $x \in \R^d$} \mper
            \numberthis \label{eq:sum-of-distinct-products}
        \end{equation*}
    \end{enumerate}

    First, since $t \leq d/\log^2 n$, by \Cref{fact:gaussian-matrix}, with probability $1 - 2^{-2t \log n}$ we have $\sigma_{\max}(A_S) \leq \sqrt{d} + O(\sqrt{t\log n}) \leq (1+o_n(1)) \sqrt{d}$.
    Then, a union bound over all $\binom{n}{t} \leq 2^{t\log n}$ choices of $S$ shows that this holds for all $S$.
    This certifies \ref{item:submatrix-singular-value} of \Cref{lem:2-to-4-excluding-top}.

    To certify \ref{item:4-norm}, first note that $\ol{T}_x = [n] \setminus T_x$ is the set of indices with the top $t$ largest values of $\angles{a_i,x}^2$, and by \ref{item:submatrix-singular-value} we have $\sum_{i\in \ol{T}_x} \angles{a_i,x}^2 = \norm{A_{\ol{T}_x} x}_2^2 \leq d (1+o_n(1))$.
    Thus, it follows that
    \begin{equation*}
        \angles{a_i,x}^2 \leq \tau \coloneqq \frac{d}{t} (1+o_n(1)) \quad \text{for all $i\in T_x$} \mper
        \numberthis \label{eq:ai-dot-x-bound}
    \end{equation*}

    The parameters of $d,n$ and $t$ satisfy $n = \frac{d^2}{t}\log^2 n$ and $t \log^6 n \leq n + \frac{d^2}{t}\log^2 n$, the requirements for \Cref{lem:sum-of-distinct-products}.
    Thus, we can certify \Cref{eq:sum-of-distinct-products} in $n^{O(t)}$ time for all $s \leq t$.

    We now proceed to bound $\sum_{i\in T_x} \angles{a_i, x}^4$.
    We will use \Cref{eq:sum-of-distinct-products,eq:ai-dot-x-bound} to bound the $t$-th power:
    \begin{equation*}
        \parens*{\sum_{i\in T_x} \angles{a_i, x}^4}^t
        = \sum_{i_1,\dots,i_t \in T_x} \prod_{k\in[t]} \angles{a_{i_k}, x}^4
        = \sum_{S \subseteq T_x} \sum_{\substack{i_1,\dots,i_t\in S: \\ \supp(i_1,\dots,i_t) = S}} \prod_{k\in[t]} \angles{a_{i_k}, x}^4
        \mcom
    \end{equation*}
    where we group the terms according to the support of the indices.
    For any $i_1,\dots,i_t \in T_x$ with $\supp(i_1,\dots,i_t) = S$, since $\angles{a_i, x}^4 \leq \tau^2$ for $i\in T_x$ (\Cref{eq:ai-dot-x-bound}), we have
    \begin{equation*}
        \prod_{k\in[t]} \angles{a_{i_k},x}^4 \leq \tau^{2(t-|S|)} \prod_{j\in S} \angles{a_j, x}^4 \mper
    \end{equation*}
    Moreover, for any $S \subseteq T_x$ of size $s \leq t$, we claim that the number of ordered indices $(i_1,\dots,i_t)$ with support $S$ is at most $t! \cdot \binom{t-1}{t-s}$.
    To see this, we can construct $(i_1,\dots,i_t)$ by first choosing $t-s$ elements from $S$ with replacement, for which there are $\binom{s + (t-s)-1}{t-s} = \binom{t-1}{t-s}$ choices (alternatively, this is the number of ways of throwing $t$ balls into $s$ bins so that no bin is non-empty), and then there are $t!$ ways to permute the indices.
    Thus,
    \begin{equation*}
    \begin{aligned}
        \parens*{\sum_{i\in T_x} \angles{a_i, x}^4}^t
        &\leq \sum_{s=1}^t \sum_{S\subseteq T_x: |S|=s} t! \binom{t-1}{t-s} \tau^{2(t-s)}  \prod_{i\in S} \angles{a_i,x}^4 \\
        &\leq \sum_{s=1}^t \frac{t!}{s!} \binom{t-1}{t-s} \tau^{2(t-s)} \cdot s! \sum_{S\subseteq [n]: |S|=s} \prod_{i\in S} \angles{a_i,x}^4 \mper \\
    \end{aligned}
    \end{equation*}
    Note that we changed the summation of $S\subseteq T_x$ to $S \subseteq [n]$ (since $T_x \subseteq [n]$ and all terms are non-negative).
    Crucially, we get an upper bound that does not depend on $T_x$, thus allowing us to use \Cref{eq:sum-of-distinct-products}.
    Since $\frac{d^2}{s}\log^2 n \geq \frac{d^2}{t} \log^2 n = n$ and $\tau = O(\frac{d}{t})$, the above is bounded by
    \begin{equation*}
    \begin{aligned}
        \sum_{s=1}^t \frac{t!}{s!} \binom{t-1}{t-s} \cdot O\parens*{\frac{d}{t}}^{2(t-s)}
        O\parens*{\frac{d^2}{s}\log^2 n}^{s}
        &\leq O\parens*{\frac{d^2}{t}\log^2 n}^t \sum_{s=1}^t \binom{t-1}{t-s} \parens*{\frac{t}{s}}^s  \\
        &\leq O(n)^t \sum_{s=1}^t \binom{t-1}{t-s} \binom{t}{s} \\
        &\leq O(n)^t \cdot 2^{2t} \mper
    \end{aligned}
    \end{equation*}
    Thus, we have certified that $\parens*{\sum_{i\in T_x} \angles{a_i, x}^4}^t \leq O(n)^t$, completing the proof.
\end{proof}

With \Cref{lem:2-to-4-excluding-top}, the proof of \Cref{thm:spread} is straightforward.

\begin{proof}[Proof of \Cref{thm:spread} from \Cref{lem:2-to-4-excluding-top}]
    Let $x\in \R^d$ be any unit vector, let $y = Ax$, and let $S \subseteq [n]$ be any subset of size $\leq \alpha n$.
    We would like to certify that $\norm{y_{\ol{S}}}_2 \geq \frac{1}{2}\norm{y}_2$.
    Let $T_x \subseteq [n]$ be as defined in \Cref{lem:2-to-4-excluding-top}, i.e., the set of indices excluding the top $n^{\eps}$ largest $\angles{a_i,x}^2$.
    By \ref{item:submatrix-singular-value} of \Cref{lem:2-to-4-excluding-top}, we can certify that $\norm{y_{\ol{T}_x}}_2^2 = \norm{A_{\ol{T}_x} x}_2^2 \leq d (1+o_n(1))$ since $\abs*{\ol{T}_x} = n^{\eps}$ by definition.
    Moreover, by \Cref{fact:gaussian-matrix} we know that $\sigma_{\min}(A) \geq \sqrt{n} - \sqrt{d} - o_n(1)$ with high probability, and since $d = o(n)$, we have $\norm{y}_2^2 \geq n (1- o_n(1))$.
    This means that $\norm{y_{\ol{T}_x}}_2^2 \leq o_n(1) \norm{y}_2^2$.

    Let $S' \coloneqq S \cap T_x = S \setminus \ol{T}_x$.
    Then, we have 
    \begin{equation*}
        \norm*{y_{S'}}_2^2 \geq \norm{y}_2^2 - \norm*{y_{\ol{S}}}_2^2 - \norm{y_{\ol{T}_x}}_2^2 \geq (1 - o_n(1)) \norm{y}_2^2 - \norm{y_{\ol{S}}}_2^2 \mper
    \end{equation*}
    Next, by \ref{item:4-norm} of \Cref{lem:2-to-4-excluding-top}, $\sum_{i\in S'} y_i^4 \leq \sum_{i\in T_x} y_i^4 \leq B n$ for some constant $B$.
    Thus, by Cauchy-Schwarz,
    \begin{equation*}
        \norm*{y_{S'}}_2^2 = \sum_{i\in S'} y_i^2 \leq \sqrt{|S'| \sum_{i\in S'} y_i^4}
        \leq \sqrt{\alpha n} \cdot \sqrt{Bn}
        \leq \sqrt{\alpha B} \cdot \norm{y}_2^2 \cdot (1+o_n(1)) \mper
    \end{equation*}
    Thus, we can set the constant $\alpha$ such that $\norm*{y_{\ol{S}}}_2^2 \geq (1-\sqrt{\alpha B} - o_n(1)) \norm{y}_2^2 \geq \frac{1}{4} \norm{y}_2^2$.
\end{proof}

\subsection{Symmetrization lowers spectral norm}
\label{sec:symmetrization-lowers-trace}

Note that $\angles{a_i, x}^4 = \angles{a_i^{\otimes 4}, x^{\otimes 4}}$, where we may view $a_i^{\otimes 4}$ as a $4$-th order tensor or a $d^2 \times d^2$ matrix.
We start by redistributing the entries of $a_i^{\otimes 4}$.

\begin{definition} \label{def:shift-entries}
    Given a $d^2 \times d^2$ matrix $A$ indexed by tuples $(j_1, j_2), (j_3, j_4) \in [d]$, we define $\shift(A)$ to be the $d^2 \times d^2$ matrix such that
    \begin{equation*}
        \shift(A)_{(j_1,j_2),(j_3,j_4)} = 
        \begin{cases}
            0 & \text{if $j_1 = j_2 \notin \{j_3, j_4\}$ or $j_3 = j_4 \notin \{j_1,j_2\}$} \\
            \frac{3}{2} A_{(j_1,j_2),(j_3,j_4)} & \text{if $j_1 \neq j_2$, $j_3 \neq j_4$, and $|\{j_1,j_2,j_3,j_4\}| \leq 3$ } \\
            A_{j_1,j_2,j_3,j_4} & \text{otherwise} \mper
        \end{cases}
    \end{equation*}
\end{definition}
For example $\shift(A)_{(1,1),(2,2)} = 0$ and $\shift(A)_{(1,2),(1,2)} = \frac{3}{2} A_{(1,2),(1,2)}$.
This is also done in \cite{BarakBHKSZ12} to remove a large rank-$1$ component in $\Sigma = \E[a_i^{\otimes 4}]$ (recall \Cref{sec:BBH}) so that $\shift(\Sigma)$ has norm $O(1)$ even though $\|\Sigma\|_2 \geq n$.
In \Cref{rem:lower-bound-without-shifting}, we will see that without this, the spectral norm bound is false.


\begin{proposition}
    For any $a, x\in \R^d$, $\angles{a,x}^4 = \angles*{\shift(a^{\otimes 4}), x^{\otimes 4}}$.
\end{proposition}
\begin{proof}
    Expanding $\angles{a, x}^4$, we get $\sum_{j_1,j_2,j_3,j_4} a_{j_1} a_{j_2} a_{j_3} a_{j_4} x_{j_1} x_{j_2} x_{j_3} x_{j_4}$, where we can group the terms according to the multiset $\{j_1,j_2,j_3,j_4\}$.
    The only entries that differ between $a^{\otimes 4}$ and $\shift(a^{\otimes 4})$ are the ones where the multiset $\{j_1,j_2,j_3,j_4\}$ is of the following form:
    \begin{itemize}
        \item $j_1 = j_2 \neq j_3 = j_4$: there are $6$ such terms in $a^{\otimes 4}$, and there are $4$ such terms in $\shift(a^{\otimes 4})$, each scaled by $\frac{3}{2}$.
        \item $j_1 = j_2 \neq j_3 \neq j_4$: there are $12$ such terms in $a^{\otimes 4}$, and there are $8$ such terms in $\shift(a^{\otimes 4})$, each scaled by $\frac{3}{2}$.
    \end{itemize}
    This shows that $\angles{a,x}^4 = \angles*{\shift(a^{\otimes 4}), x^{\otimes 4}}$.
\end{proof}

To prove \Cref{lem:sum-of-distinct-products}, we start by writing $t!\sum_{|S|=t} \prod_{i\in S} \angles*{a_i, x}^4 = t!\sum_{|S|=t} \prod_{i\in S} \angles*{\shift(a_i^{\otimes 4}), x^{\otimes 4}}$ as a quadratic form:
\begin{equation*}
    t!\sum_{S \subseteq [n]: |S|=t} \prod_{i\in S} \angles*{\shift(a_i^{\otimes 4}),x^{\otimes 4}}
    = \angles*{t!\sum_{S\subseteq[n]:|S|=t} \bigotimes_{i\in S} \shift(a_i^{\otimes 4}), x^{\otimes 4t}}
    = (x^{\otimes 2t})^\top M (x^{\otimes 2t}) \mcom
    \numberthis \label{eq:M-quad-form}
\end{equation*}
where we view $x^{\otimes 2t}$ as a vector of dimension $d^{2t}$, and
\begin{equation*}
    M \coloneqq t! \sum_{S\subseteq[n]:|S|=t} \bigotimes_{i\in S} \shift(a_i^{\otimes 4}) \in \R^{d^{2t} \times d^{2t}} \mper
    \numberthis \label{eq:M-matrix}
\end{equation*}
Here, we view $\shift(a_i^{\otimes 4})$ as a $d^{2} \times d^{2}$ matrix.

The natural approach to bound \Cref{eq:M-quad-form} is to bound the spectral norm of $M$.
Unfortunately, it can be shown that $\|M\|_2 \geq \Omega(d)^{2t}$.
As explained in \Cref{sec:upper-bound-Pt}, we resolve this by exploiting the symmetry of $x^{\otimes 2t}$.
Let $\Pi \in \R^{d^{2t} \times d^{2t}}$ be any permutation matrix that maps an index $J = (j_1,j_2,\dots,j_{2t})\in [d]^{2t}$ to $(j_{\pi(1)}, j_{\pi(2)}, \dots, j_{\pi(2t)})$ for some permutation $\pi$ over $2t$ elements.
Then, we have that $\Pi x^{\otimes 2t} = x^{\otimes 2t}$.

We will use $\Permutation{2t}$ to denote the collection of all such matrices $\Pi$.
It follows that
\begin{equation*}
    (x^{\otimes 2t})^\top M (x^{\otimes 2t}) = (x^{\otimes 2t})^\top \wt{M} (x^{\otimes 2t})
    \leq \norm{\wt{M}}_2 \mcom
    \quad \text{where } \wt{M} = \E_{\Pi,\Pi'\sim \Permutation{2t}}[\Pi M \Pi'] \mper
    \numberthis \label{eq:wt-M-matrix}
\end{equation*}

Our key technical result is that the symmetrization lowers the spectral norm by a factor of roughly $t^t$, and we prove it using the trace moment method, a standard technique for upper bounding spectral norm of matrices.

\begin{lemma} \label{lem:trace-bound}
    Let $a_1,\dots,a_n \in \R^d$ be random vectors with independent entries such that $\E[a_{ij}^k] = 0$ for all odd $k$, $\E[a_{ij}^2] \leq 1$, $\E[a_{ij}^4] \leq \mu_4$, and $|a_{ij}| \leq \sqrt{C\log n}$ almost surely for $\mu_4 \geq 1$ and constant $C > 0$.
    Let $t \leq d \leq n$ be integers such that $t \log^6 n \leq \mu_4 n + \frac{d^2}{t} \log^2 n$.
    Let $\wt{M}$ be the $d^{2t} \times d^{2t}$ matrix defined in \Cref{eq:wt-M-matrix}.
    Let $\ell \leq \log n$.
    Then,
    \begin{equation*}
        \E_A\bracks*{\tr(\wt{M}^{\ell})} \leq d^{2t}\cdot O\parens*{\mu_4 n + \frac{d^2}{t}\log^2 n}^{t\ell} \mper
    \end{equation*}
\end{lemma}

In \Cref{rem:trace-lower-bound}, we will show that the upper bound in \Cref{lem:trace-bound} is tight up to log factors.

We can immediately complete the proof of \Cref{lem:sum-of-distinct-products} using \Cref{lem:trace-bound}.
\begin{proof}[Proof of \Cref{lem:sum-of-distinct-products} from \Cref{lem:trace-bound}]
    First, since $a_1, \dots, a_n \in \R^d$ are vectors with $\calN(0,1)$ entries, we have that $\E[a_{ij}^4] = 3$, and moreover, with probability $1 - \frac{1}{\poly(n)}$ we have that $|a_{ij}| \leq \sqrt{C \log n}$ for some constant $C > 0$ for all $i\in [n]$, $j\in [d]$.
    Thus, we can now condition on this event; the conditioned variables still satisfy $\E[a_{ij}^k] = 0$ for odd $k$, $\E[a_{ij}^2] \leq 1$ and $\E[a_{ij}^4] \leq 3$, i.e., the conditions in \Cref{lem:trace-bound}.

    Set $\ell = \log n$.
    By \Cref{lem:trace-bound} and Markov's inequality, with probability $1 - \exp(-t\ell) \geq 1 - n^{-t}$, we have that $\tr(\wt{M}^{\ell}) \leq d^{2t} \cdot O(n + \frac{d^2}{t}\log^2 n)^{t\ell}$.
    Then, $\norm{\wt{M}}_2 \leq \tr(\wt{M}^\ell)^{1/\ell} \leq O(n + \frac{d^2}{t}\log^2 n)^{t}$, using the fact that $d^{2t/\ell} \leq 2^{2t}$.

    We can construct the matrix $\wt{M} \in \R^{d^{2t} \times d^{2t}}$ and calculate its spectral norm in $n^{O(t)}$ time.
    Since $(x^{\otimes 2t})^{\top} \wt{M} (x^{\otimes 2t}) \leq \norm{\wt{M}}_2 \cdot \norm{x^{\otimes 2t}}_2^2$ and $\norm{x^{\otimes 2t}}_2 = 1$ for all unit vectors $x$, this completes the proof.
\end{proof}

\subsection{Encoding of walks in the trace}
\label{sec:encoding}

In this section, we prove \Cref{lem:trace-bound}.
Expanding $\E_A \tr(\wt{M}^\ell)$ gives
\begin{equation*}
    \E_A \tr\parens*{\wt{M}^\ell}
    = \E_A \tr \E_{\Pi_1,\Pi_1', \Pi_2, \Pi_2',\dots, \Pi_{\ell}'} \bracks*{\Pi_1 M \Pi_1' \Pi_2 M \Pi_2' \cdots \Pi_{\ell} M \Pi_{\ell}'} \mper
\end{equation*}
We will use $\bpi$ to denote the sequence of permutation matrices $(\Pi_1, \Pi_1', \Pi_2, \Pi_2', \dots, \Pi_\ell') \in (\Permutation{2t})^{2\ell}$.

The trace power of a matrix is often viewed as a sum of closed walks.
Indeed, $\tr(M^{\ell}) = \sum_{I_1, I_2,\dots, I_{\ell} \in [d]^{2t}} M_{I_1,I_2} M_{I_2,I_3} \cdots M_{I_\ell,I_1}$, which is a weighted sum of length-$\ell$ closed walks $I_1 \to I_2 \to \cdots \to I_\ell \to I_1$ on the indices.
With a specific permutation $\bpi$, it is then a sum of closed walks where the indices are permuted accordingly at each step.
See \Cref{fig:structure-examples} for an example.

\begin{definition}[Structure] \label{def:structure}
    We will use a graph with two types of vertices (circle and square) to represent a walk of length $\ell$ in the trace. The graph is divided into $\ell$ blocks, each with $t$ square vertices in the middle and $4t$ edges between square and circle vertices.
    We use $\bpi = (\Pi_1, \Pi_1', \dots, \Pi_{\ell}, \Pi_{\ell}') \in (\Permutation{2t})^{2\ell}$ to denote the \emph{structure} of the graph, where the edges in block $k$ are connected according to $\Pi_k, \Pi_k'$.

    For example, \Cref{fig:structure-examples} shows two structures with $t = 3$.
    In \Cref{fig:structure}, all permutations are identity.
    In \Cref{fig:structure-permuted}, we have permutations $\pi_1,\pi_1',\dots, \pi_\ell, \pi_\ell'$; for example, on the left side of the first block, $\pi_1(1) = 3, \pi_1(2) =2, \pi_1(3) = 5$ and so on, i.e., circle $i$ is connected to square $\ceil{\pi_1(i) / 2}$.
\end{definition}

\begin{figure}[ht!]
    \centering
    \begin{subfigure}[b]{0.45\textwidth}
        \includegraphics[width=\textwidth]{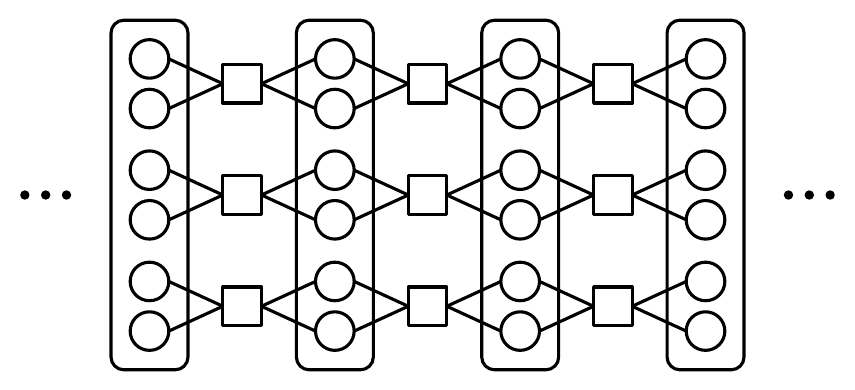}
        \caption{Identity permutations.}
        \label{fig:structure}
    \end{subfigure}
    \quad
    \begin{subfigure}[b]{0.45\textwidth}
        \includegraphics[width=\textwidth]{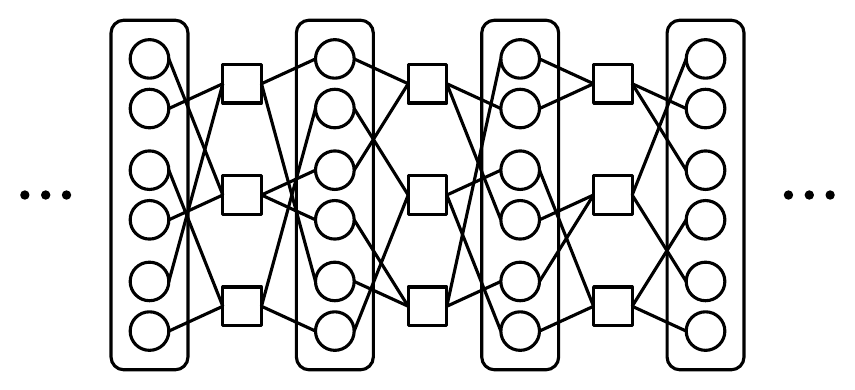}
        \caption{Different permutations.}
        \label{fig:structure-permuted}
    \end{subfigure}
    \caption{Examples of different structures of walks in the trace.
    }
    \label{fig:structure-examples}
\end{figure}

\begin{definition}[Valid labeling of a structure]
\label{def:valid-labeling}
    A \emph{labeling} $\sigma$ is a map that maps circle vertices to $[d]$ and square vertices to $[n]$.
    We say that a labeling $\sigma$ is \emph{valid} for a structure $\bpi$ if
    \begin{enumerate}[(1)]
        \item All labeled edges (i.e., elements in $[n] \times [d]$) appear even number of times.
        \label{item:edge-appears-twice}

        \item The $t$ square vertices in each block receive distinct labels.
        \label{item:t-distinct-squares}

        \item For each square vertex, if the 2 incident circle vertices on the left (resp.\ right) have the same labels $j\in [d]$, then there must be at least one circle vertex on the right (resp.\ left) that is labeled $j$.
        \label{item:square-neighbors}
    \end{enumerate}
    Moreover, we define $\edgefactor(\sigma,\bpi)$ to be the product of factors from the labeled edges, where a labeled edge appearing $k$ times gets a factor of $\E[a^k]$.
\end{definition}

Requirement~\ref{item:edge-appears-twice} is because $\E[a_{ij}^k]= 0$ for all odd $k$, and because of this, labelings that violate requirement~\ref{item:edge-appears-twice} automatically have $\edgefactor(\sigma,\bpi) = 0$.
Requirement~\ref{item:t-distinct-squares} is by definition of the $M$ matrix (\Cref{eq:M-matrix}), which is a sum over $S \subseteq [n]$.
Finally, we can impose requirement~\ref{item:square-neighbors} on the labelings because $\shift(a_i^{\otimes 4})_{(j_1,j_1), (j_2,j_3)} = 0$ for $j_1 \notin \{ j_2, j_3\}$.
Indeed, recall from \Cref{def:shift-entries} that $\shift(a_i^{\otimes 4})_{(j_1,j_1), (j_2,j_3)}$ is nonzero only if one or both of $j_2, j_3$ equal $j_1$.

\begin{remark}[Lower bound on $\E \tr(\wt{M}^\ell)$] \label{rem:trace-lower-bound}
    We claim that $\E \tr(\wt{M}^\ell) \geq \Omega(\mu_4 n + \frac{d^2}{t})^{t\ell}$, thus the upper bound in \Cref{lem:trace-bound} is tight up to log factors.

    Fix a subset $S \subseteq [n]$ and consider the labelings where the circle vertices are all distinct (i.e., $I_1, I_2, \dots, I_{\ell} \in [d]^{2t}$ collectively have distinct indices) while the square vertices in each block is $S$ (up to ordering).
    There are $\binom{d}{2t\ell} (2t\ell)! \approx d^{2t\ell}$ choices for the circle vertices (as $d \gg 2t\ell$).
    Since each circle vertex is distinct, the two adjacent square vertices (in neighboring blocks) must be the same so that each edge appears twice.
    Note that the square vertices in each block can be reordered.
    Thus, for the circle-to-square structure in each block, viewing it as a matching between $2t$ circle vertices, there is exactly one matching that results in an even labeling.
    The probability is $\frac{1}{(2t-1)!!} \geq t^{-t}$.
    Thus, these labelings contribute a lower bound of $d^{2t\ell} \cdot t^{-t\ell} = (\frac{d^2}{t})^{t\ell}$.

    Next, consider the labelings where all circle vertices are the same while the square vertices are distinct.
    There are $\approx n^{t\ell}$ choices for the square vertices.
    Moreover, there are $t\ell$ distinct edges, each appearing $4$ times, hence giving a factor of $\mu_4^{t\ell}$.
    Thus, these contribute a lower bound of $(\mu_4 n)^{t\ell}$.
\end{remark}

\begin{remark}[A higher lower bound on $\E \tr(\wt{M}^\ell)$ without requirement~\ref{item:square-neighbors}] \label{rem:lower-bound-without-shifting}
    Fix the identity permutation structure (\Cref{fig:structure}).
    Consider the labeling such that for each square vertex, the two circle vertices on each side get the same labels (thus violating requirement~\ref{item:square-neighbors}).
    Then, we assign distinct labels to all square vertices and all pairs of circle vertices, for which there are $\approx d^{t\ell} n^{t\ell}$ choices.
    After permutation, i.e., randomizing the structure, with probability $(\frac{1}{(2t-1)!!})^{2\ell} \geq t^{-2t\ell}$ each pair of circle vertices are still paired, in which case all edges appear twice, satisfying requirement~\ref{item:edge-appears-twice}.
    Thus, this gives a lower bound of $\Omega(\frac{dn}{t^2})^{t\ell}$, which is much larger than the target bound $O(n + \frac{d^2}{t})^{t\ell}$ in \Cref{lem:trace-bound}.
\end{remark}

Now, we can bound $\E \tr(\wt{M}^{\ell})$ as follows,

\begin{equation*}
    \E_{A} \tr\parens*{\wt{M}^\ell} \leq 2^{O(t\ell)} \cdot \E_{\bpi} \sum_{\sigma: \text{ labeling}} \1(\text{$\sigma$ valid for $\bpi$}) \cdot \edgefactor(\sigma,\bpi) \mper
\end{equation*}
Here, the $2^{O(t\ell)}$ factor takes care of the fact that $\shift(a^{\otimes 4})$ has some entries scaled by $3/2$.

The crucial step in bounding the above is defining an encoding of the labelings.
Following the terminology of \cite{Tao12}, we will refer to a circle-to-square or square-to-circle step in the walk (structure) as a \emph{leg}, and we will refer to an element in $[n] \times [d]$ as an \emph{edge}.
In other words, a labeling of a leg in the structure specifies an edge.
An example is shown in \Cref{fig:encoding-examples}.

\begin{remark}[Encoding requirement]
\label{rem:structure}
    We would like an encoding scheme that uses as few bits as possible such that
    (1) every valid labeling for a given structure can be encoded,
    and (2) a labeling is uniquely determined given an encoding and a structure (if decoding succeeds).
    In particular, the encoding itself does not need to identify the structure; instead, the decoding algorithm will take both the encoding and the structure to decode.
    Thus, when provided with the same encoding and two distinct structures, the decoding algorithm can successfully output two different labelings, each valid for its respective structure.
\end{remark}

\begin{definition}[Encoding of a labeling] \label{def:encoding}
    We define our encoding of a valid labeling for a structure $\bpi$ as follows.
    First, specify a starting index $J \in [d]^{2t}$.
    Then, for each of the $4t\ell$ legs ($2t$ circle to square and $2t$ square to circle in each block), specify a \emph{type} in $\{\Fresh, \Paired, \Ret, \NI, \Highmult\}$ along with additional information for decoding (namely, \textbf{destination} and \textbf{return labels} depending on the type):
    \begin{itemize}
        \item $\Fresh$ (Fresh) leg: the destination is a new vertex (not seen before), thus creating a new edge, marked as \emph{active}.
        An $\Fresh$ leg from circle to square must be paired with a $\Paired$ leg as described next.
        The new vertex is specified as follows:
        \begin{itemize}
            \item Circle to square: an element in $[n]$.
            \item Square to circle: an element in $[d]$ and a \emph{bucket index} $b\in [\log n]$ for the circle vertex.
            The bucket index indicates that the vertex will have degree between $2^b$ and $2^{b+1}$ in total.
        \end{itemize}

        \item $\Paired$ (Paired) leg: it must be from a circle to a new square vertex and must be paired with an $\Fresh$ leg.
        There is no need to specify a vertex label because the vertex is determined by the paired $\Fresh$ leg.
        The $\Paired$ legs are further split into two types:
        \begin{itemize}
            \item $\Paired_1$:
            the paired $\Fresh$ and $\Paired$ legs come from \emph{different} circle vertices.
            It marks the edge as active.

            \item $\Paired_2$:
            the paired $\Fresh$ and $\Paired$ legs come from the \emph{same} circle vertex.
            By requirement~\ref{item:square-neighbors} of \Cref{def:valid-labeling}, one or two of the next legs leaving the square vertex must be an $\Ret$ going back to the circle vertex.
        \end{itemize}
        In addition, specify a return label in $\{1,2\}$.

        \item $\Ret$ (Return) leg: it traverses an incident active edge and marks it as \emph{closed}.
        The incident edge will be chosen (during decoding) using return labels from previous legs.

        \item $\NI$ (Non-innovative) leg: the destination is an old vertex (seen before), but the edge is new and is marked as active.
        \begin{itemize}
            \item Circle to square: specify a previous square vertex (a label in $[t\ell]$), and then specify a return label in $[4\ell]$ (because square labels in a layer are distinct).
            \item Square to circle: pick a bucket index $b\in [\log n]$ and choose a vertex from the bucket by specifying an element in $[\frac{4t\ell}{2^b}]$.
            Then, specify a return label in $[2^{b+1}]$.
        \end{itemize}

        \item $\Highmult$ (High-multiplicity): it traverses an incident \emph{closed} edge and marks it as \emph{active}, indicating that the edge is traversed an odd number of times.
        \begin{itemize}
            \item Circle to square: specify an incident edge (a label in $[4t\ell]$), and then specify a return label in $[\ell]$.
            \item Square to circle: same as $\NI$.
        \end{itemize}
    \end{itemize}
    Finally, let $\xi$ be an encoding with $\highmult$ number of $\Highmult$ legs and $\paired_2$ number of $\Paired_2$ legs.
    With a slight abuse of notation, we define $\edgefactor(\xi)$ as $\mu_4^{\paired_2} (C\log n)^{\highmult}$, where $\mu_4$ and $C$ are the constants such that our random variables satisfy $\E[a_{ij}^4] \leq \mu_4$ and $|a_{ij}| \leq \sqrt{C \log n}$ almost surely.
\end{definition}

See \Cref{fig:encoding-examples} for examples of an encoding.
We note that the terms ``fresh'', ``return'', ``non-innovative'' and ``high-multiplicity'' in \Cref{def:encoding} are adopted from \cite{Tao12}\footnote{In \cite{HsiehKPX23}, a non-innovative leg is called a ``surprise''.}
(see \Cref{sec:trace-method}).
The $\Paired$ (Paired) type is introduced to handle the specific structure of our matrix, namely that there are two legs leading to a square.

\begin{remark}[$\Paired_1$ and $\Paired_2$ legs] \label{rem:P-legs}
    As each $\Paired_1$ leg creates a new edge, it should be viewed as an $\Fresh$ leg except that it does not need to specify the destination.

    On the other hand, a $\Paired_2$ leg should be viewed as an $\Highmult$ leg which is closed immediately when departing from the square.
    When two circle vertices with the same label $j\in [d]$ go to a square $i\in [n]$ via an $\Fresh$ and $\Paired_2$ leg, the edge $(\square{i}, \circle{j})$ is traversed twice.
    However, by requirement~\ref{item:square-neighbors}, one or two of the next legs leaving the square vertex must be $\Ret$ legs going back to $j$, so after this the edge is traversed either $3$ or $4$ times.
    Thus, the $\Paired_2$ leg is effectively an $\Highmult$ leg that traverses the edge the \emph{third} time, which introduces an edge factor of $\E[a_{ij}^4] \leq \mu_4$.
    See \Cref{rem:lower-bound-without-shifting} for a lower bound without requirement~\ref{item:square-neighbors}.
\end{remark}

It is helpful for readers to keep in mind that eventually the dominating terms in the trace calculation will be encodings with mostly $\Fresh$ (or $\{\Fresh,\Paired\}$ if leading to a square) and $\Ret$ legs.
In this case, most $\Ret$ legs are ``forced'', i.e., there is only one edge to choose.
In fact, the only times an $\Ret$ leg is ``unforced'' at a vertex $v$ (i.e., it needs to choose between multiple incident edges) are when there were previous $\NI$ or $\Highmult$ legs going back to $v$.
When this happens, the destination of the $\Ret$ leg is determined by one return label among those placed by previous $\NI$ and $\Highmult$ legs.
We now need to show that the return labels in our encoding are sufficient to guide the ``unforced'' return legs.

\paragraph{Unforced return legs.}
Again following the terminology in \cite{Tao12}, we say that an $\Ret$ (return) leg from a vertex $v$ is \emph{forced} if it marks the final visit to $v$ in the walk.
Consequently, following this leg, all edges incident to $v$ are closed for the remainder of the walk.
Every circle vertex has exactly $1$ forced return,
and every square vertex has $2$ (the two legs leaving $v$ in its final appearance).
All other return legs are \emph{unforced}.

In \Cref{def:encoding}, each $\NI$ and $\Highmult$ leg arriving at $v$ provides 1 return label for $v$.
The next lemma shows that this is sufficient for all unforced returns departing from $v$.

\begin{lemma}[Bounding unforced returns] \label{lem:pur}
    For each vertex $v$, the number of unforced return legs departing from $v$ is at most the number of $\NI$ and $\Highmult$ legs arriving at $v$.
    Consequently, the return labels in \Cref{def:encoding} are sufficient for all unforced returns.
\end{lemma}
\begin{proof}
    Let $\fresh_{\into}, \n_{\into}, \highmult_{\into}, \ret_{\into}$ be the number of $\Fresh, \NI, \Highmult, \Ret$ steps arriving at $v$,
    and $\fresh_{\out}, \n_{\out}, \highmult_{\out}, \ret_{\out}$ be the number of $\Fresh, \NI, \Highmult, \Ret$ steps departing from $v$.
    For circle vertices, there may be $\paired_{\out}$ out-going $\Paired$ legs.
    We know that (1) the number of in-going and out-going edges are the same: $\fresh_{\into} + \n_{\into} + \highmult_{\into} + \ret_{\into} = \fresh_{\out} + \paired_{\out} + \n_{\out} + \highmult_{\out} + \ret_{\out}$,
    and (2) all active edges need to be closed: $\ret_{\into} + \ret_{\out} = \fresh_{\into}  + \fresh_{\out} + \paired_{\out} + \n_{\into} + \n_{\out} + \highmult_{\into} + \highmult_{\out}$.
    Combining the two, we get
    \begin{equation*}
        \ret_{\into} + \ret_{\out} = 2(\fresh_{\into} + \n_{\into} + \highmult_{\into}) + \ret_{\into} - \ret_{\out}
        \implies \ret_{\out} = \fresh_{\into} + \n_{\into} + \highmult_{\into} \mper
    \end{equation*}
    For circle vertices, $\fresh_{\into} = 1$ (the first arrival at $v$), while there is $1$ forced return from $v$.
    So, there are $\n_{\into} + \highmult_{\out}$ number of unforced returns from $v$.

    Now, if $v$ is a square vertex, there is exactly $1$ incoming $\Fresh$ and $1$ $\Paired$ leg (the first arrival at $v$; all subsequent arrivals are $\NI$, $\Highmult$ or $\Ret$ by definition), meaning $\fresh_{\into} + \paired_{\into} = 2$.
    The same calculation shows that $\ret_{\out} = \fresh_{\into} + \paired_{\into} + \n_{\into} + \highmult_{\into}$.
    There are $2$ forced returns from $v$, so again there are $\n_{\into} + \highmult_{\out}$ number of unforced returns from $v$.
\end{proof}

We remark that this generalized definition of forced/unforced returns was used in \cite{HsiehKPX23} as well to bound walks that have complicated structures.
In fact, \Cref{lem:pur} is essentially the Potential-Unforced-Return (PUR) factor as defined in \cite{HsiehKPX23}.

\subsection{Decoding}
\label{sec:decoding}

Given an encoding and a structure, we will decode the labeling step by step from left to right.
Each edge is either active or closed during the decoding process, and all edges must be closed in the end (so that requirement~\ref{item:edge-appears-twice} of \Cref{def:valid-labeling} is satisfied).

\begin{mdframed}
    \begin{algorithm}[Decoding algorithm $\Decode$]
    \label{alg:decode}
    \mbox{}
      \begin{description}
      \item[Input:] A structure $\bpi \in (\Permutation{2t})^{2\ell}$, an encoding $\xi$.

      \item[Output:] A labeling of the structure, or $\FAIL$.

      \item[Operation:] \mbox{}
        We will label the walk (the graph determined by $\bpi$) one vertex at a time.
        It is convenient to view the algorithm as dynamically maintaining a separate bipartite graph between labeled square and circle vertices (elements in $[n]$ and $[d]$ respectively), where a leg in the walk may create a new vertex and/or edge, and each edge has either ``active'' or ``closed'' status.
        Each vertex stores a list of return labels.
        Moreover, each circle vertex is assigned a bucket index $b$.
        \begin{itemize}
            \item
            \begin{itemize}
                \item $\{\Fresh,\Paired\}$ legs from circle to square: label the square vertex with the specified label, and mark the edges as active.
                Place down the return label from the $\Paired$ leg.

                \item $\Fresh$ leg from square to circle: label the circle vertex with the specified label, and mark the edge as active.
                Then, label the circle vertex with the specified bucket index.

                \item $\NI$ or $\Highmult$ leg: label the vertex by the specified previous vertex, and mark the edge as active.
                Then, place down the return label.

                \item $\Ret$ leg: use a return label to choose an active edge incident to the current vertex to close.
                In the special case that it is a square to circle leg and the previous visit to the square is via $\{\Fresh, \Paired_2\}$ from two circle vertices both labeled $j$, then simply set $j$ as the destination.
                If there is no return label, there must be only one choice (otherwise output $\FAIL$).
            \end{itemize}

            \item For every square vertex, unless it is reached via $\{\Fresh,\Paired\}$ legs, it gets two labels from the two legs.
            Output $\FAIL$ if the two labels differ.
            Also output $\FAIL$ if the labels assigned to the $t$ square vertices in the same block are not distinct.

        \end{itemize}
      \end{description}
    \end{algorithm}
\end{mdframed}

\begin{figure}[ht!]
    \centering
    \begin{subfigure}[b]{0.45\textwidth}
        \includegraphics[width=\textwidth]{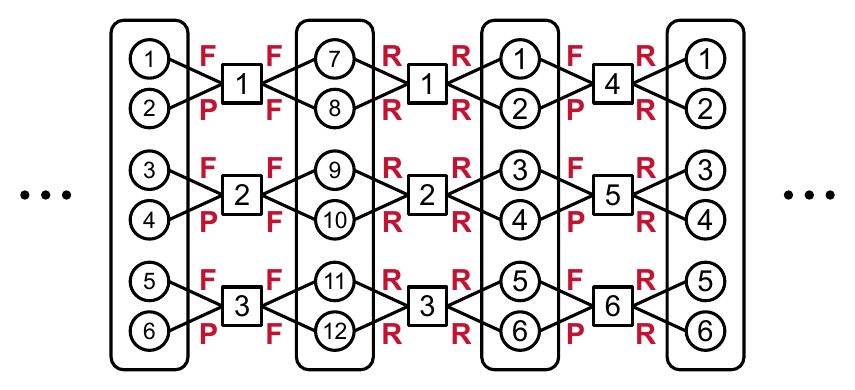}
        \caption{Decoding succeeds.}
        \label{fig:encoding-succeeds}
    \end{subfigure}
    \quad
    \begin{subfigure}[b]{0.45\textwidth}
        \includegraphics[width=\textwidth]{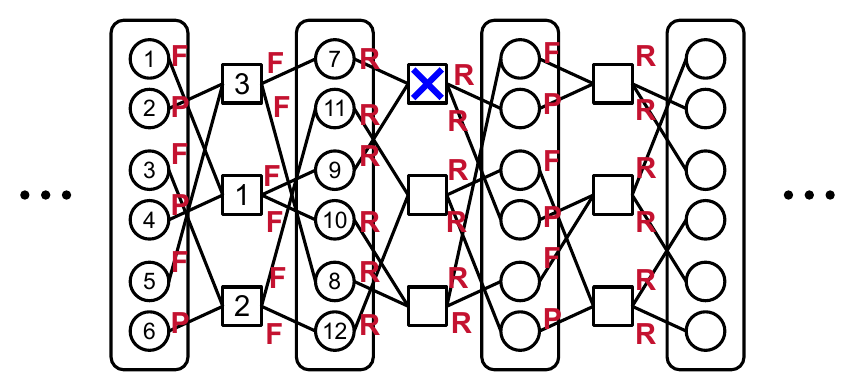}
        \caption{Decoding fails.}
        \label{fig:encoding-permuted}
    \end{subfigure}
    \caption{An example where we decode the same encoding with different structures.
    }
    \label{fig:encoding-examples}
\end{figure}

\Cref{fig:encoding-examples} shows an example where we decode the same encoding with different structures.
The first one (\Cref{fig:encoding-succeeds}) succeeds; in fact, for the part shown in \Cref{fig:encoding-succeeds}, the resulting labeling satisfies all requirements of a valid labeling (\Cref{def:valid-labeling}).
The second one fails because for the square vertex (marked with an ``x''), the two legs give conflicting labels: the first one labels it $3$ to close the edge $(\square{3}, \circle{7})$ while the second one labels it $1$ to close the edge $(\square{1}, \circle{9})$.

It is crucial that any valid labeling is captured by some encoding, so that we can bound the sum of labelings by the sum of encodings.

\begin{lemma} \label{lem:labeling-to-encoding}
    For any labeling $\sigma$ which is valid for some structure $\bpi \in (\Permutation{2t})^{2\ell}$, there is an encoding $\xi$ such that $\Decode(\xi,\bpi)$ successfully outputs $\sigma$.
    Moreover, $\edgefactor(\xi) \geq \edgefactor(\sigma,\bpi)$.
\end{lemma}
\begin{proof}
    Given a valid labeling for a structure, the $4t\ell$ leg types and the specification of vertices (a new vertex for $\Fresh$, a previous location for $\NI$, and an incident edge for $\Highmult$) are straightforward to determine.
    For circle vertices, the bucket index $b\in [\log n]$ is specified such that the total number of edges incident to the vertex is between $2^b$ and $2^{b+1}$.
    Here $b \leq \log n$ since $4t\ell \ll 2^{\log n} = n$.
    In particular, the number of circle vertices in bucket $b$ is at most $\frac{4t\ell}{2^b}$.

    The return labels are also straightforward.
    As shown in \Cref{lem:pur}, one return label for each $\NI$ or $\Highmult$ is enough to guide all unforced returns.
    Thus, we can specify the return labels arriving at a vertex $v$ according to the ordering in which active edges incident to $v$ are closed.
    A circle vertex with bucket index $b$ has at most $2^{b+1}$ incident edges, whereas a square vertex has at most $4\ell$ incident edges (due to requirement~\ref{item:t-distinct-squares} of \Cref{def:valid-labeling} that the square vertices in each block are distinct).
    For a square vertex $v$, two departing $\Ret$ legs are taken to close the final two active edges during the last visit to $v$. The return label for the $\Paired$ leg (an element in $\{1,2\}$) that initially creates $v$ is decided by the ordering of these final two $\Ret$ legs.
    Therefore, the ranges of return labels as defined in the encoding are sufficient.

    For the edge factors, consider an edge traversed $2k \geq 2$ times in the walk.
    If the edge is created by an $\{\Fresh, \Paired_2\}$ pair, then as discussed in \Cref{rem:P-legs}, we may view the $\Paired_2$ leg as the third time the edge is traversed.
    In this case, the edge will be traversed by $1$ $\Fresh$, $1$ $\Paired_2$, $(k-2)$ $\Highmult$ and $k$ $\Ret$ legs.
    Then, the edge factor of $\xi$ from this edge will be $\mu_4 \cdot (C \log n)^{k-2}$.
    This is a valid upper bound as $\E[a^{2k}] \leq \E[a^4] \cdot (\sqrt{C \log n})^{2k-4} \leq \mu_4 \cdot (C\log n)^{k-2}$.

    On the other hand, in other cases the edge will be traversed by $1$ $\Fresh$, $(k-1)$ $\Highmult$ and $k$ $\Ret$ legs, giving an edge factor of $(C\log n)^{k-1}$.
    This is also a valid upper bound as $\E[a^{2k}] \leq \E[a^2] \cdot (\sqrt{C\log n})^{2k-2} \leq (C\log n)^{k-1}$.
    Therefore, $\edgefactor(\xi)$ as defined in \Cref{def:encoding} is a valid upper bound on $\edgefactor(\sigma,\bpi)$.
\end{proof}

The next lemma states two simple but important relationships between the number of legs of each type in a valid encoding.

\begin{lemma} \label{lem:encoding-params}
    Let $\xi$ be an encoding such that $\Decode(\xi,\bpi)$ succeeds for some $\bpi$.
    Let $\fresh_s, \paired_s, \n_s, \highmult_s, \ret_s$ be the number of $\Fresh, \Paired, \NI, \Highmult, \Ret$ legs from circle to square, let $\fresh_c, \n_c, \highmult_c, \ret_c$ be the number of $\Fresh, \NI, \Highmult, \Ret$ legs from square to circle, and let $\n = \n_s + \n_c$, $\highmult = \highmult_s + \highmult_c$ and $\ret = \ret_s + \ret_c$.
    Then, we must have
    \begin{enumerate}[(1)]
        \item $\ret = 2\fresh_s + \fresh_c + \n + \highmult = 2t\ell$.

        \item $\n_s + \highmult_s + \ret_s = \fresh_c + \n + \highmult$.
    \end{enumerate}
\end{lemma}
\begin{proof}
    First, recall from the encoding (\Cref{def:encoding}) that each $\Paired$ must be paired with an $\Fresh$ from circle to square, so $\paired_s = \fresh_s$.
    Next, each $\Fresh$, $\Paired$, $\NI$ and $\Highmult$ leg creates an active edge that needs to be closed by an $\Ret$ leg, thus
    \begin{equation*}
        \ret = \fresh + \paired_s + \n + \highmult = 2\fresh_s + \fresh_c + \n + \highmult \mper
        \numberthis \label{eq:r}
    \end{equation*}
    Next, the total number of edges is $4t\ell$, so $2\fresh_s + \fresh_c + \n + \highmult + \ret = 4t\ell$, which means that $\ret = 2\fresh_s + \fresh_c + \n + \highmult = 2t\ell$.

    Moreover, as the number of edges from circle to square is the same as that from square to circle, we have
    \begin{equation*}
        2\fresh_s + \n_s + \highmult_s + \ret_s = \fresh_c + \n_c + \highmult_c + \ret_c \mper
    \end{equation*}
    Combined with \Cref{eq:r}, we get $\ret_s + \ret_c = 2(\fresh_c + \n_c + \highmult_c) + (\ret_c - \ret_s)$, which means $\ret_s = \fresh_c + \n_c + \highmult_c$.
    Thus, $\n_s + \highmult_s + \ret_s = \fresh_c + \n + \highmult$.
\end{proof}

We next prove the crucial lemma upper bounding the success probability of decoding over random permutations of the structure.

\begin{lemma}[Decoding success probability] \label{lem:decoding-success-probability}
    Let $\xi$ be an encoding, let $\fresh_c$ be the number of $\Fresh$ legs from square to circle, and let $\n, \highmult$ be the number of $\NI, \Highmult$ legs.
    Then,
    \begin{equation*}
        \Pr_{\bpi\sim (\Permutation{2t})^{2\ell}} \bracks*{\text{$\Decode(\xi,\bpi)$ succeeds} }
        \leq \parens*{\frac{t}{e}}^{-\frac{1}{2} (\fresh_c + \n + \highmult)} \mper
    \end{equation*}
\end{lemma}
\begin{proof}
    Let $\fresh_s, \paired_s, \n_s, \highmult_s, \ret_s$ be the number of $\Fresh, \Paired, \NI, \Highmult, \Ret$ legs from circle to square, and let $\fresh_c, \n_c, \highmult_c, \ret_c$ be the number of $\Fresh, \NI, \Highmult, \Ret$ legs from square to circle.
    Consider the circle-to-square part of each block, and view the structure as a perfect matching between $2t$ circle vertices.
    The $\Fresh$ and $\Paired$ legs must be paired and the $t$ square vertices must receive distinct labels.
    Moreover, all $\NI, \Highmult, \Ret$ legs from circle to square provide labels, and recall that \Cref{alg:decode} fails if any square vertex receives two different labels.
    Therefore, there is exactly one way to pair up all the $\NI, \Highmult, \Ret$ legs.

    Suppose there are $i \leq t$ number of $\Fresh$, $i$ number of $\Paired$ (because $\Fresh$ and $\Paired$ must be paired), and $2t-2i$ other legs.
    There are $(2t-1)!! \geq (\frac{2t}{e})^t$ number of ways to group them into pairs.
    However, there are $i!$ ways to pair the $\Fresh$ and $\Paired$ legs, and there is a unique way to pair the other legs.
    Thus, via Stirling's approximation, the probability that a grouping is valid is
    \begin{equation*}
        \frac{i!}{(2t-1)!!} \leq i! \parens*{\frac{2t}{e}}^{-t}
        \leq \parens*{\frac{t}{e}}^{-(t-i)} \mper
    \end{equation*}
    Note that $t-i$ is $1/2$ of the number of $\NI, \Highmult$ and $\Ret$ legs from circle to square in the block.
    Thus, overall, the probability of success is upper bounded by $(t/e)^{-\frac{1}{2} (\n_s + \highmult_s + \ret_s)}$.
    By \Cref{lem:encoding-params}, we have $\n_s + \highmult_s + \ret_s = \fresh_c + \n + \highmult$, thus completing the proof.
\end{proof}

\subsection{Proof of \texorpdfstring{\Cref{lem:trace-bound}}{Lemma~\ref{lem:trace-bound}}}
\label{sec:proof-of-trace-bound}

We can finally finish the proof of \Cref{lem:trace-bound}.

\begin{proof}[Proof of \Cref{lem:trace-bound}]
    We start with writing the trace as a sum of encodings:
    \begin{equation*}
    \begin{aligned}
        \E_{A} \tr\parens*{\wt{M}^\ell}
        &\leq \E_{\bpi\sim (\Permutation{2t})^{2\ell}} \sum_{\sigma: \text{ labeling}} \1(\text{$\sigma$ valid for $\bpi$}) \cdot \edgefactor(\sigma,\bpi) \\
        &\leq \E_{\bpi\sim (\Permutation{2t})^{2\ell}} \sum_{\sigma: \text{ labeling}} \sum_{\xi: \text{ encoding}} \1(\Decode(\xi,\bpi)=\sigma \text{ \& valid for $\bpi$}) \cdot \edgefactor(\xi) \\
        &\leq \sum_{\xi: \text{ encoding}} \Pr_{\bpi\sim (\Permutation{2t})^{2\ell}} \bracks*{\text{$\Decode(\xi,\bpi)$ succeeds} } \cdot \edgefactor(\xi) \mper
    \end{aligned}
    \end{equation*}
    Here, the second inequality follows from \Cref{lem:labeling-to-encoding} that every labeling $\sigma$ has an encoding $\xi$ such that $\Decode(\xi, \bpi)$ successfully outputs $\sigma$, and that $\edgefactor(\sigma,\bpi) \leq \edgefactor(\xi)$.

    We bound the number of encodings as follows:
    \begin{itemize}
        \item For each $\{\Fresh,\Paired\}$ leading to a square, there are $n$ choices for the new vertex, and $2$ choices for the return label of $\Paired$. In total, $2n$ choices.

        \item For each $\Fresh$ from square to circle, there are $d$ choices for the new vertex and $\log n$ choices for the bucket index.
        In total, $d \log n$ choices.

        \item For each $\NI$ from circle to square, there are $t\ell$ choices for a previous square vertex, and $4\ell$ choices for the return label.
        In total, $4t\ell^2$ choices.

        \item For each $\NI$ or $\Highmult$ from square to circle, there are $\log n$ choices to choose a bucket index $b$, and $\frac{4t\ell}{2^b}$ choices for a vertex in bucket $b$.
        Then, there are $2^{b+1}$ choices for a return label.
        In total, $\log n \cdot \frac{4t\ell}{2^b} \cdot 2^{b+1} = 8t\ell \log n$ choices.

        \item For each $\Highmult$ from circle to square, there are $4t\ell$ choices to choose an incident edge, and $4\ell$ choices for a return label.
        In total, $16t\ell^2$ choices.

        \item \textbf{Edge factor:} Each $\Highmult$ gets a $C\log n$ factor, and each $\Paired_2$ gets a $\mu_4$ factor.
        Recall that $\mu_4$ and $C$ are the constants such that $\E[a_{ij}^4] \leq \mu_4$ and $|a_{ij}| \leq \sqrt{C \log n}$ almost surely.
    \end{itemize}

    Fix a leg type pattern and starting index in $[d]^{2t}$ (labels for the first $2t$ circle vertices in the structure), and let $\fresh_s, \paired_s, \n_s, \highmult_s, \ret_s$ be the number of $\Fresh, \Paired, \NI, \Highmult, \Ret$ legs from circle to square, and let $\fresh_c, \n_c, \highmult_c, \ret_c$ be the number of $\Fresh, \NI, \Highmult, \Ret$ legs from square to circle.
    Moreover, let $\n = \n_s + \n_c$ and $\highmult = \highmult_s + \highmult_c$.
    Then, the total encodings, weighted by $\edgefactor$, is bounded by
    \begin{equation*}
    \begin{gathered}
        n^{\fresh_s} (2\mu_4)^{\paired_s} (d\log n)^{\fresh_c} (4t\ell^2)^{\n_s} (8t\ell \log n)^{\n_c} (16Ct\ell^2 \log n)^{\highmult_s} (8Ct\ell \log^2 n)^{\highmult_c} \\
        \leq (2\mu_4 n)^{\fresh_s} (d\log n)^{\fresh_c} (8t\log^2 n)^{\n} (16Ct\log^3n)^{\highmult} \mcom
    \end{gathered}
    \end{equation*}
    where we use the fact that $\paired_s = \fresh_s$, $\mu_4\geq 1$ and $\ell \leq \log n$.

    Next, by \Cref{lem:decoding-success-probability}, 
    $\Pr_{\bpi}[\Decode(\xi,\bpi) \text{ succeeds}] \leq (t/e)^{-\frac{1}{2}(\fresh_c + \n + \highmult)}$.
    Thus, we get
    \begin{equation*}
        (2 \mu_4 n)^{\fresh_s} \parens*{\sqrt{\frac{e}{t}} d\log n}^{\fresh_c} O\parens*{\sqrt{t} \log^3 n}^{\n + \highmult}
        \leq O\parens*{\mu_4 n + \frac{d^2}{t}\log^2 n}^{\fresh_s + \frac{1}{2}\fresh_c} O\parens*{\sqrt{t} \log^3 n}^{\n + \highmult}
    \end{equation*}
    By \Cref{lem:encoding-params}, $2\fresh_s + \fresh_c = 2t\ell - \n - \highmult$.
    Thus, as long as $t \log^6 n \leq \mu_4 n + \frac{d^2}{t} \log^2 n$, the above is bounded by $O(\mu_4 n + \frac{d^2}{t}\log^2 n)^{t\ell}$.

    Finally, there are at most $2^{O(t\ell)}$ leg type patterns, and $d^{2t}$ choices for the starting index $J \in [d]^{2t}$ in the encoding.
    This gives the final bound of $d^{2t} \cdot O(\mu_4 n + \frac{d^2}{t}\log^2 n)^{t\ell}$.
\end{proof}

\section{Planted Sparse Vector}
\label{sec:planted-sparse-vector}


We first restate the planted sparse vector model.

\begin{model*}[Restatement of \Cref{model:planted-dist}]
    Fix an unknown unit vector $v\in \R^n$, and let $d \leq n \in \N$.
    Let $\wt{A}$ be a random $n \times d$ matrix sampled as follows: (1) let $A$ be the random matrix such that the first column is $v$ and the other $d-1$ columns are i.i.d.\ $\calN(0, \frac{1}{n}\Id_n)$ vectors;
    (2) let $R \in \R^{d \times d}$ be an arbitrary unknown rotation matrix;
    (3) set $\wt{A} = A R$.

    The task is that given $\wt{A}$, output a unit vector $\wh{v} \in \R^n$ such that $\angles{\wh{v}, v}^2 \geq 1 - o(1)$.
\end{model*}

We first characterize the sparse vectors that we will plant in \Cref{model:planted-dist}.
Motivated by the definition of well-spread subspaces (\Cref{def:spread-subspace}), we define the notion of compressible vectors as follows,

\begin{definition}[Compressible vector]
    Let $\rho,\gamma \in (0,1)$.
    We say that a vector $v \in \R^n$ is $(\rho,\gamma)$-compressible if 
    \begin{equation*}
        \max_{S\subseteq [n]: |S| \leq \rho n} \norm*{v_S}_2 = (1-\gamma) \norm{v}_2 \mper
    \end{equation*}
\end{definition}

In other words, $v$ is compressible if the top $\rho n$ coordinates of $v$ contains $1-\gamma$ fraction of its $\ell_2$-mass.

Next, recall our notation for the elementary symmetric polynomial (of the 4th powers):
\begin{equation*}
    Q_t(z) \coloneqq t! \sum_{S\subseteq [n]: |S| = t} \prod_{i\in S} z_i^4 \mper
\end{equation*}

We now state our main result.

\begin{theorem}[Formal version of \Cref{thm:planted-main}] \label{thm:planted-sparse-vector}
    There is an absolute positive constant $C$ such that for every $t \leq d \leq n \in \N$ such that $\frac{\rho d^2}{n} \log^C n \leq t \leq \frac{d}{\log^C n}$ and $\rho \in (0,\frac{1}{\log^C n})$,
    there is a randomized algorithm with running time $2^{\wt{O}(t)}$ with the following guarantee:
    Given $\wt{A} \in \R^{n\times d}$ drawn from \Cref{model:planted-dist} with a hidden unit vector $v \in \R^n$ such that $v$ is $(\rho,1/2)$-compressible and $Q_t(v)^{1/t} \geq (1 - \frac{1}{\log^C n}) \|v\|_4^4$, the algorithm outputs a unit vector $\wh{v}\in \colspan(\wt{A})$ such that $\angles{\wh{v}, v}^2 \geq 1 - o(1)$ with probability $1 - o(1)$ over the randomness of the algorithm and the input.
\end{theorem}

To interpret the parameters in \Cref{thm:planted-sparse-vector}, consider $d = O(n^{\frac{1+\eps}{2}})$ and $t = \wt{O}(n^{\eps})$ (same as in \Cref{thm:spread}), and the sparsity of  $v$ is $\rho = 1/ \polylog(n)$.
Then, we can approximately recover $v$ in $2^{\wt{O}(n^{\eps})}$ time.
This establishes the same trade-off between the dimension and runtime as our certification algorithm (\Cref{thm:certification}).

Our algorithm is described in \Cref{alg:planted-sparse-vector}, and the proof of \Cref{thm:planted-sparse-vector} is completed at the end of \Cref{sec:algo-planted}.

\begin{remark}[Assumptions on $v$] \label{rem:assumption-on-v}
    The assumption that $Q_t(v)^{1/t} \geq (1 - o(1)) \|v\|_4^4$ is not standard.
    In particular, this does not hold for any vector $v$ whose mass concentrates on $<t$ coordinates.
    However, for such vectors, one can simply brute-force search over all support of size $\leq \wt{O}(t)$, which takes $2^{\wt{O}(t)}$ time.

    On the other hand, most distributions of sparse vectors used in previous works satisfy that $Q_t(v)^{1/t} \geq (1-o(1)) \|v\|_4^4$ (when the sparsity is not small enough to brute-force search over the support).
    These include the noiseless and noisy Rademacher-Bernoulli distribution (\Cref{def:nBR}) considered in \cite{MaoW21,ChendO22,DiakonikolasK22,ZadikSWB22} and the Rademacher-Gaussian distribution considered in \cite{MaoW21}.

    We note that the algorithms of \cite{HopkinsSSS16,MaoW21} require very minimal assumptions on $v$.
    In particular, the algorithm of \cite{MaoW21} succeeds as long as $\abs*{\|v\|_4^4 - \frac{3}{n}} \geq \Omega(\frac{1}{n})$ and $\frac{1}{\|v\|_4^4} \frac{d}{n^{3/2}} \ll 1$,\footnote{Here $\frac{3}{n}$ is exactly $\E[\|g\|_4^4]$ for a Gaussian vector $g \sim \calN(0, \frac{1}{n}\Id_n)$. There is also a requirement on $\|v\|_{\infty}/\|v\|_4^2$ which we omit, as it is satisfied in most cases. See \cite{MaoW21} for the full statement.}
    so when $\|v\|_4^4 \approx \frac{1}{\rho n}$, this is equivalent to $\rho d \ll \sqrt{n}$.
    We leave as future work the identification of a minimal analytic assumption on $v$ that ensures successful recovery in the subexponential-time regime.

\end{remark}

\subsection{Preliminaries for the planted sparse vector model}

It is easy to see that a compressible vector $v$ has large $4$-norm:

\begin{lemma} \label{lem:compressible-4-norm}
    A $(\rho,\gamma)$-compressible vector $v \in \R^n$ satisfies $\|v\|_4^4 \geq \frac{(1-\gamma)^4}{\rho n} \|v\|_2^2$.
\end{lemma}
\begin{proof}
    Let $S^* \subseteq [n]$ be such that $|S^*| \leq \rho n$ and $\|v_{S^*}\|_2 \geq (1-\gamma)\|v\|_2$.
    Then, by Cauchy-Schwarz, $(1-\gamma)^4 \|v\|_2^4 \leq \|v_{S^*}\|_2^4 \leq |S| \cdot \|v_S\|_4^4 \leq \rho n \cdot \|v\|_4^4$.
\end{proof}

The following lemma is also straightforward given the standard singular value concentration bounds on Gaussian matrices (\Cref{fact:gaussian-matrix}).

\begin{lemma} \label{lem:A-norm}
    Let $d \leq n/\polylog(n)$.
    Let $A \in \R^{n\times d}$ be drawn from \Cref{model:planted-dist} with an arbitrary unit vector $v$.
    Then, with probability $1 - o(1)$, we have $1 - o(1) \leq \sigma_{\min}(A) \leq \sigma_{\max}(A) \leq 1 + o(1)$.
\end{lemma}
\begin{proof}
    Recall that the first column of $A$ is $v$, and let $\ol{A} \in \R^{n\times (d-1)}$ be the rest of the matrix with i.i.d.\ $\calN(0,1/n)$ entries.
    By \Cref{fact:gaussian-matrix}, we know that $\sigma_{\max}(\ol{A}) \leq 1 + o(1)$ with high probability.
    Moreover, since $v$ is a unit vector, $\ol{A}^\top v$ is distributed as a $(d-1)$-dimensional vector with $\calN(0,1/n)$ entries, thus with high probability we have $\|\ol{A}^\top v\|_2^2 \leq O(d/n) \leq o(1)$.

    For any unit vector $x = (x_1, \ol{x})$, we have $Ax = x_1 v + \ol{A} \ol{x}$, thus
    \begin{align*}
        \|Ax\|_2^2 &= x_1^2 \|v\|_2^2 + \norm*{\ol{A} \ol{x}}_2^2 + 2x_1 v^\top \ol{A} \ol{x}
        \leq x_1^2 \|v\|_2^2 + \norm*{\ol{A} \ol{x}}_2^2 + 2|x_1| \|\ol{x}\|_2 \norm*{\ol{A}^\top v}_2 \\
        &\leq (1+o(1)) \parens*{x_1^2 + \norm*{\ol{x}}_2^2}
        = (1+o(1)) \norm{x}_2^2 \mper
    \end{align*}
    Similarly, the minimum singular value lower bound follows from the fact that  $\sigma_{\min}(\ol{A}) \geq 1-o(1)$
    and $\|Ax\|_2^2 \geq x_1^2 \|v\|_2^2 + \|\ol{A} \ol{x}\|_2^2 - 2|x_1| \|\ol{x}\|_2 \norm{\ol{A}^\top v}_2$.
\end{proof}

The following lemma is important in our algorithm to identify a ``good'' vector to output.

\begin{lemma} \label{lem:selecting-the-vector}
    Let $n, d\in \N$ such that $d \leq n/\polylog(n)$, let $\rho \leq 1/\polylog(n)$, and let $\gamma,\delta \in (0,1)$.
    Let $v \in \R^n$ be a $(\rho,\gamma)$-compressible unit vector, and let $A \in \R^{n\times d}$ be drawn from \Cref{model:planted-dist} with planted vector $v$.
    Then, with probability $1- o(1)$ over $A$, the following holds:
    let $y\in \R^d$ be any unit vector and let $\wh{v} = A y$,
    \begin{enumerate}[(1)]
        \item If $\|y - e_1\|_2 \leq \delta$ or $\|y + e_1\|_2 \leq \delta$, then $\max_{S: |S| \leq \rho n} \|\wh{v}_S\|_2 \geq 1 - \gamma - (1+o(1))\delta$.
        \label{item:close-to-e1-implies-sparse}

        \item If $\|y - e_1\|_2 \geq \delta$ and $\|y + e_1\|_2 \geq \delta$, then $\max_{S:|S|\leq \rho n} \|\wh{v}_S\|_2 \leq (1 - \frac{1}{2}\delta^2)(1 - \gamma) + \frac{1}{\polylog(n)}$.
        \label{item:far-from-e1-implies-spread}
    \end{enumerate}
\end{lemma}
\begin{proof}
    Let $S^* = \argmax_{S:|S| \leq \rho n} \|v_S\|_2$, and recall that $Ae_1 = v$.

    Suppose $\|y - e_1\|_2 \leq \delta$.
    Then, $\|\wh{v} - v\|_2 = \|A(y - e_1)\|_2 \leq \sigma_{\max}(A) \delta \leq (1+o(1))\delta$ with high probability by \Cref{lem:A-norm}.
    By the triangle inequality,
    $\|\wh{v}_{S^*}\|_2 \geq \|v_{S^*}\|_2 - \|v_{S^*} - \wh{v}_{S^*}\|_2 \geq 1 - \gamma - (1+o(1))\delta$ since $\|v_{S^*} - \wh{v}_{S^*}\|_2 \leq \|v - \wh{v}\|_2 \leq (1+o(1))\delta$.
    The same is true if we flip the sign of $y$ so that $\wh{v}$ is close to $-v$.
    This proves the first statement.

    Suppose $\|y - e_1\|_2 \geq \delta$ and $\|y+e_1\|_2 \geq \delta$.
    Denote $y = (y_1, \ol{y})$, where $\ol{y}\in \R^{d-1}$ and $\norm{\ol{y}}_2^2 = 1 - y_1^2$.
    Then, we have $\delta^2 \leq (1-y_1)^2 + \norm{\ol{y}}_2^2 = 2 - 2y_1$ and $\delta^2 \leq (1+y_1)^2 + \norm{\ol{y}}_2^2 = 2 + 2y_1$, which means that $|y_1| \leq 1 - \delta^2/2$.
    For any $S \subseteq [n]$ with $|S| \leq \rho n$, we have $\sigma_{\max}(\ol{A}_S) \leq \wt{O}(\sqrt{d/n} + \sqrt{\rho}) \leq 1/\polylog(n)$ by \Cref{fact:gaussian-matrix} and a union bound over all $S$ (here we need $\rho, d/n \leq \frac{1}{\polylog(n)}$).
    Then, since $\wh{v}_S = y_1 v_S + \ol{A}_S \ol{y}$, by the triangle inequality, $\|\wh{v}_S\|_2 \leq |y_1| \norm{v_S}_2 + \norm*{\ol{A}_S \ol{y}}_2 \leq (1 - \frac{1}{2}\delta^2) \norm{v_{S^*}}_2 + \frac{1}{\polylog(n)}$, which proves the second statement.
\end{proof}

\subsection{Background on Sum-of-Squares}
\label{sec:sos}

In this section, we give an overview of the Sum-of-Squares (SoS) framework.
We refer the reader to the monograph~\cite{FlemingKP19} and the lecture notes~\cite{BarakS16} for a detailed exposition of the SoS method and its usage in algorithm design. 

\paragraph{Pseudo-distributions.}
Pseudo-distributions are generalizations of probability distributions and are represented by their \emph{pseudo-expectation} operators.
Formally, a degree-$t$ pseudo-distribution $\mu$ over variables $x_1, , \dots, x_n$ corresponds to a linear operator $\pE_{\mu}$ that maps polynomials of degree $\leq t$ to real numbers and satisfies $\pE_\mu[1]=1$ and $\pE_{\mu}[p^2] \geq 0$ for every polynomial $p(x_1,\ldots, x_n)$ of degree $\leq t/2$.

Let $\calA = \{f_1 \geq 0, f_2 \geq 0, \dots, f_m\geq 0\}$ be a system of $m$ polynomial inequality constraints.
We say that \emph{$\mu$ satisfies the system of constraints $\calA$ at degree $t$} if for every sum-of-squares polynomial $h$ and any $T \subseteq [m]$ such that $\deg(h) + \sum_{i\in T} \deg(f_i) \leq t$, $\pE_{\mu} [h \cdot \prod _{i\in T} f_i] \ge 0$.
In particular, if $\calA$ contains an \emph{equality} constraint $f = 0$, then $\pE_{\mu}[f\cdot p] = 0$ for any polynomial $p$ with $\deg(f) + \deg(p) \leq t$.
Specific to our application, we say that a pseudo-distribution $\mu$ satisfies the \emph{unit sphere} constraints if $\pE_{\mu}[ \norm{x}_2^2 \cdot p(x)] =  \pE_{\mu}[p(x)]$ for every $p$ of degree $\leq t-2$. 

Unlike true distributions, there is a $n^{O(t)}$-time weak separation oracle for degree-$t$ pseudo-distributions, which allows us to efficiently optimize over pseudo-distributions that satisfy a given set of polynomial constraints (approximately) via the ellipsoid method --- this is called the Sum-of-Squares algorithm.
Therefore, given a polynomial $p$ (with the $\ell_1$-norm of the coefficients being $\norm{p}_1$) over $x_1,\dots,x_n$, a degree-$t$ pseudo-distribution satisfying the unit sphere constraint that maximizes $p$ within an additive $\varepsilon \norm{p}_1$ error can be found in time $n^{O(t)} \polylog(n/\varepsilon)$.

\paragraph{Sum-of-Squares proofs.}
Let $f_1,f_2,\dots,f_m$ and $g$ be multivariate polynomials in $x$.
A \emph{sum-of-squares} proof that the constraints $\calA = \{f_1 \geq 0,\dots, f_m \geq 0\}$ imply $g \geq 0$ consists of sum-of-squares polynomials $(p_S)_{S\subseteq[m]}$ such that $g(x) = \sum_{S\subseteq[m]} p_S(x) \prod_{i\in S} f_i(x)$.
The \emph{degree} of such an SoS proof equals the maximum of the degree of $p_S \prod_{i \in S} f_i$ over all $S$ appearing in the sum above.
We write
\begin{equation*}
    \calA \sststile{t}{x} \{g(x) \geq 0\}
\end{equation*}
where $t$ is the degree of the SoS proof.

The following fact is the crucial connection between SoS proofs and pseudo-distributions.

\begin{fact}
Suppose $\calA \sststile{t}{x} \{g(x) \geq 0\}$ for some polynomial constraints $\calA$ and a polynomial $g$.
Let $\mu$ be any pseudo-distribution of degree $\geq t$ satisfying $\calA$. Then, $\pE_\mu[g] \geq 0$.
\end{fact}

In other words, for polynomials $f, g$, in order to prove that $\pE_{\mu}[f] \leq \pE_{\mu}[g]$ for any degree-$t$ pseudo-distribution $\mu$ satisfying $\calA$, it suffices to show an SoS proof that $\calA \sststile{t}{x} \{f(x) \leq g(x)\}$.



We will need the following SoS versions of Cauchy-Schwarz and \Holder's inequalities.

\begin{fact}[Fact 3.9 and 3.11 in \cite{ODonnelW13}] \label{fact:holder-deg4}
    For any $\delta > 0$,
    \begin{itemize}
        \item $\sststile{2}{x,y} \braces*{ xy \leq \frac{\delta}{2}x^2 + \frac{1}{2\delta} y^2 }$.

        \item $\sststile{4}{x,y} \braces*{ x^3 y \leq \frac{3\delta}{4}x^4 + \frac{1}{4\delta^3} y^4 }$.
    \end{itemize}
\end{fact}

Using \Cref{fact:holder-deg4}, we can prove the following lemma, which is basically the SoS version of Jensen's inequality applied to the $x^4$ function.
\begin{lemma} \label{lem:jensen}
    For any $\delta > 0$,
    \begin{equation*}
        \sststile{4}{x,y} \braces*{ (x+y)^4 \leq (1+\delta)^3 \parens*{x^4 + \frac{y^4}{\delta^3}} } \mper
    \end{equation*}
\end{lemma}
\begin{proof}
    We expand $(x+y)^4 = x^4 + 4x^3 y + 6x^2 y^2 + 4 x y^3 + y^4$ and apply \Cref{fact:holder-deg4} to the middle 3 terms:
    $(x+y)^4 \leq x^4 + (3\delta x^4 + \frac{1}{\delta^3} y^4) + (3\delta^2 x^4 + \frac{3}{\delta^2} y^4) + (\delta^3 x^4 + \frac{3}{\delta} y^4) + y^4 = (1+\delta)^3 x^4 + (1 + \frac{1}{\delta})^3 y^4$.
\end{proof}

\paragraph{Rounding.}
In the following, we state an important technique for rounding pseudo-distributions from \cite{BarakKS15}.

\begin{lemma}[Theorem 5.1 of \cite{BarakKS15}] \label{lem:rounding-powers}
    For every even $t\in \N$ and $\eps \in (0,1)$, there exists a randomized algorithm with running time $n^{O(t)}$ and success probability $2^{-t/\poly(\eps)}$ for the following problem:
    Fix an unknown unit vector $q\in \R^n$. Given a degree-$t$ pseudo-distribution $\mu$ over $x\in \R^n$ satisfying the constraint $\|x\|_2^2 = 1$ such that $\pE_{\mu}[\angles{q, x}^t] \geq e^{-\eps t}$, output a unit vector $\wh{x} \in \R^n$ with $\angles{q, \wh{x}} \geq 1 - O(\eps)$.
\end{lemma}


\subsection{Algorithm for planted sparse vector}
\label{sec:algo-planted}

\begin{mdframed}
    \begin{algorithm}[Recover hidden sparse vector]
    \label{alg:planted-sparse-vector}
    \mbox{}
      \begin{description}
      \item[Input:] A matrix $\wt{A} \in \R^{n\times d}$ drawn from \Cref{model:planted-dist}, parameters $t\in \N$ and $\rho \in (0,1)$.

      \item[Output:] A unit vector $\wh{v} \in \colspan(\wt{A})$.

      \item[Operation:] \mbox{}
        \begin{enumerate}
            \item Solve the degree-$4t$ SoS relaxation of the following program over $x\in \R^{d}$;
            \begin{equation*}
            \begin{aligned}
                \max\quad & \wt{P}_t(x) \coloneqq t! \sum_{S\subseteq [n]: |S|=t} \prod_{i\in S} \angles*{\wt{a}_i,x}^4 \\
                \text{s.t.} \quad  & \|x\|_2^2 = 1
            \end{aligned}
            \end{equation*}
            where $\wt{a}_1,\dots, \wt{a}_n$ are the rows in $\wt{A}$.

            \item Repeat the algorithm of \Cref{lem:rounding-powers} $N = 2^{\wt{O}(t)}$ times and obtain unit vectors $\wh{x}^{(1)}, \wh{x}^{(2)},\dots,\wh{x}^{(N)}$.

            \item Let $\wh{v}^{(i)} = \wt{A} \wh{x}^{(i)}$, and let $j = \argmax_{i\in[N]} \max_{S \subseteq [n]: |S|\leq \rho n} \norm{\wh{v}^{(i)}_S}_2$, i.e., the vector whose top $\rho n$ entries have the largest norm.
            Output $\wh{v}^{(j)} / \|\wh{v}^{(j)}\|_2$.
        \end{enumerate}
      \end{description}
    \end{algorithm}
\end{mdframed}

The main ingredient in the analysis is \Cref{lem:sum-of-distinct-products}.
Below, we state it as a degree-$4t$ SoS proof.

\begin{lemma}[SoS version of \Cref{lem:sum-of-distinct-products}] \label{lem:sum-of-distinct-products-sos}
    Let $t \leq d \leq n$ be integers such that $t \log^6 n \leq n + \frac{d^2}{t} \log^2 n$.
    Let $A \sim \calN(0, \frac{1}{n})^{n\times d}$ with rows $a_1,\dots,a_n \in \R^d$.
    Then, with probability $1 - \frac{1}{\poly(n)}$ over $a_1,\dots,a_n$,
    \begin{equation*}
        \sststile{4t}{y}
        \braces*{
        P_t(y) \coloneqq t!\sum_{S \subseteq [n]: |S|=t} \prod_{i\in S} \angles*{a_i, y}^4 \leq O\parens*{\frac{1}{n} + \frac{d^2}{n^2 t}\log^2 n}^t \norm{y}_2^{4t} } \mper
    \end{equation*}
\end{lemma}
Note the scaling of $1/n$ here because we assume the entries to be $\calN(0,1/n)$ as opposed to $\calN(0,1)$ in \Cref{lem:sum-of-distinct-products}.

Using \Cref{lem:sum-of-distinct-products-sos}, we prove our main lemma.

\begin{lemma} \label{lem:large-y1}
    Assume the same setting as \Cref{thm:planted-sparse-vector}.
    Let $r_1 \in \R^d$ be the first row of the unknown rotation matrix $R$.
    Further, let $\mu$ be the degree-$4t$ pseudo-distribution over $x\in \R^n$ with the unit sphere constraint that maximizes $\wt{P}_t(x)$.
    Then, $\pE_{\mu}[\angles{r_1, x}^{2t}] \geq e^{- O(t/\log n)}$.
\end{lemma}

\begin{proof}
    Recall from \Cref{model:planted-dist} that $\wt{A} = AR$ and $v$ is the first column of $A$.
    Denote $y = Rx$ such that $\angles{\wt{a}_i,x} = \angles{a_i, Rx} = \angles{a_i, y}$ and $\wt{P}_t(x) = P_t(y)$.
    Note also that $y_1 = \angles{r_1,x}$.
    Our goal is to prove that $\pE_{\mu}[y_1^{2t}]$ is large.

    For each row $a_i$ in $A$, we write $a_i = (v_i, \ol{a}_i)$, and for simplicity denote $w_i \coloneqq \angles{\ol{a}_i, \ol{y}}$ such that $\angles{a_i, y} = y_1 v_i + w_i$.
    We first apply \Cref{lem:jensen} with $\delta \coloneqq \frac{1}{\log^3 n}$ to each $(y_1 v_i + w_i)^4$:
    \begin{equation*}
        \sststile{4t}{y}
        \braces*{ P_t(y)
        \leq t! (1+\delta)^{3t} \sum_{S\subseteq[n]: |S|=t} \prod_{i\in S} \parens*{y_1^4 v_i^4 + \frac{w_i^4}{\delta^3}} } \mper
    \end{equation*}
    Next, we expand the above by splitting $S$ into disjoint $S'$ and $S''$:
    \begin{equation*}
        t! (1+\delta)^{3t} \sum_{s=0}^t \frac{y_1^{4(t-s)}}{\delta^{3s}}  \sum_{\substack{S' \subseteq[n]: \\ |S'|=t-s}} \prod_{i\in S'} v_i^4 \sum_{\substack{S'' \subseteq [n] \setminus S': \\ |S''|=s}} \prod_{j\in S''} w_j^4 \mper
        \numberthis \label{eq:split-S}
    \end{equation*}

    Recall that $Q_s(z) \coloneqq s! \sum_{S\subseteq[n]: |S|=s} \prod_{i\in S} z_i^4$.
    For any $s > 0$ and $S' \subseteq [n]$ with $|S'| = t-s$,
    by \Cref{lem:sum-of-distinct-products-sos}, with high probability we have
    \begin{equation*}
        \sststile{y}{4s}
        \braces*{ s! \sum_{\substack{S'' \subseteq[n] \setminus S': \\ |S''|=s}} \prod_{j\in S''} w_j^4
        \leq Q_s(w)
        \leq O\parens*{\frac{1}{n} + \frac{d^2}{n^2 s} \log^2 n}^{s} \cdot \norm*{\ol{y}}_2^{4s} } \mper
    \end{equation*}
    Note that in the first inequality, we upper bound the summation over $S'' \subseteq [n] \setminus S'$ by the summation over $S \subseteq [n]$, resulting in $Q_s(w)$.

    Next, $(t-s)!\sum_{|S'| = t-s} \prod_{i\in S'} v_i^4 = Q_{t-s}(v) \leq \|v\|_4^{4(t-s)}$.
    Moreover, by \Cref{lem:compressible-4-norm}, we have $\|v\|_4^4 \geq \Omega(\frac{1}{\rho n})$ since $v$ is a $(\rho, 1/2)$-compressible unit vector.
    Then, from \Cref{eq:split-S} we have
    \begin{align*}
        \sststile{4t}{y}\ 
        \Bigg \{ P_t(y)
        &\leq (1+\delta)^{3t} \sum_{s=0}^{t} \binom{t}{s} \frac{y_1^{4(t-s)}}{\delta^{3s}} \cdot Q_{t-s}(v) Q_s(w) \\
        &\leq (1+\delta)^{3t} \sum_{s=0}^{t} \binom{t}{s} y_1^{4(t-s)} \|v\|_4^{4(t-s)} \cdot O\parens*{\frac{1}{\delta^3 n} + \frac{d^2 \log^2 n}{\delta^3 n^2 s}}^s \norm*{\ol{y}}_2^{4s} \\
        &\leq (1+\delta)^{3t} \norm{v}_4^{4t} \sum_{s=0}^t \binom{t}{s} y_1^{4(t-s)} O\parens*{\frac{\rho}{\delta^3} + \frac{\rho d^2 \log^2 n}{\delta^3 n s}}^s \norm*{\ol{y}}_2^{4s} \Bigg\} \mper
        \numberthis \label{eq:Pt-ub}
    \end{align*}

    We now take the pseudo-expectation $\pE_{\mu}$ of both sides above.
    For simplicity, we denote $\alpha \coloneqq \frac{\rho d^2 \log^2 n}{\delta^3 n t} \leq \frac{1}{\log^C n}$ for some large enough constant $C$ (since $t \geq \frac{\rho d^2}{n} \polylog(n)$).
    Moreover, $\mu$ is the pseudo-distribution that maximizes $\pE_\mu[\wt{P}_t(x)]$, and in particular the distribution supported on $x = r_1$ is feasible, so $\pE_{\mu}[\wt{P}_t(x)] \geq \wt{P}_t(r_1) = P_t(e_1) = Q_t(v)$.
    Moreover, by assumption we have $Q_t(v)^{1/t} \geq (1-\eta) \|v\|_4^4$ with $\eta = \frac{1}{\log^C(n)} \leq \delta$.
    Thus, since $\|\ol{y}\|_2^2 = 1-y_1^2$ and  $\eta \leq \delta \leq o(1)$, from \Cref{eq:Pt-ub} it follows that
    
    \begin{equation*}
        e^{-O(\delta)t} \leq \frac{(1-\eta)^t}{(1+\delta)^{3t}}  \leq \sum_{s=0}^{t} \binom{t}{s} O\parens*{\frac{\rho}{\delta^3} + \frac{\alpha t}{s}}^s \pE_{\mu} \bracks*{y_1^{4(t-s)} (1-y_1^2)^{2s}} \mper
    \end{equation*}
    Thus, there must be an $s^* \in \{0,1,\dots, t\}$ such that
    \begin{equation*}
        \binom{t}{s^*} O\parens*{\frac{\rho}{\delta^3} + \frac{\alpha t}{s^*}}^{s^*} \pE_{\mu} \bracks*{y_1^{4(t-s^*)} (1-y_1^2)^{2s^*}} \geq \frac{1}{t} e^{-O(\delta)t} \geq e^{- O(\delta) t} \mper
        \numberthis \label{eq:one-s-large}
    \end{equation*}
    Here we use the fact that $\delta t \gg \log t$.

    Let $\beta \coloneqq \frac{1}{\log^2 n} \ll 1/2$.
    If $s^* \leq \beta t$, then $\binom{t}{s^*} O(\frac{\rho}{\delta^3} + \frac{\alpha t}{s^*})^{s^*} \leq t^{O(s^*)} \leq e^{O(\beta t\log n)}$ since $t \leq n$.
    Then, as $4(t-s^*) \geq 2t$, we have
    \begin{equation*}
        \pE_{\mu}[y_1^{2t}]
        \geq \pE_{\mu} \bracks*{y_1^{4(t-s^*)} (1-y_1^2)^{2s^*}}
        \geq e^{-O(\delta+ \beta \log n)t} \geq e^{-O(t/\log n)} \mper
    \end{equation*}

    On the other hand, we claim that $s^*$ cannot be larger than $\beta t$.
    If $s^* > \beta t$, then
    \begin{equation*}
        \binom{t}{s^*} O\parens*{\frac{\rho}{\delta^3} + \frac{\alpha t}{s^*}}^{s^*}
        \leq O\parens*{\frac{1}{\beta} \parens*{\frac{\rho}{\delta^3} + \frac{\alpha}{\beta}}}^{s^*}
        \leq e^{-\beta t}
        < e^{- O(\delta) t} \mcom
    \end{equation*}
    since our parameters satisfy $\rho, \alpha \leq \frac{1}{\log^{C} n}$ for some large enough constant $C$ so that $\frac{1}{\beta}(\frac{\rho}{\delta^3} + \frac{\alpha}{\beta}) \ll 1$, and the last inequality follows from $\delta = \frac{1}{\log^3 n} \ll \beta$.
    This contradicts \Cref{eq:one-s-large}.
    Thus, it must be that $s^* \leq \beta t$, and we have $\pE_\mu[y_1^{2t}] \geq e^{-O(t/\log n)}$, completing the proof.
\end{proof}

With \Cref{lem:selecting-the-vector,lem:large-y1} and \Cref{lem:rounding-powers} (the rounding algorithm of \cite{BarakKS15}) in hand, we can now prove that \Cref{alg:planted-sparse-vector} succeeds in recovering $v$, completing the proof of \Cref{thm:planted-sparse-vector}.

\begin{proof}[Proof of \Cref{thm:planted-sparse-vector}]
    Let $r_1\in \R^d$ be the first row of the unknown rotation matrix $R$ from \Cref{model:planted-dist}, and note that $r_1$ is a unit vector.
    By \Cref{lem:large-y1}, the pseudo-distribution $\mu$ obtained from step (1) of \Cref{alg:planted-sparse-vector} satisfies that $\pE_{\mu}[\angles{r_1,x}^{2t}] \geq e^{-O(t/\log n)}$.

    By repeating the algorithm in \Cref{lem:rounding-powers} $2^{\wt{O}(t)}$ times, with high probability at least one of the unit vectors $\wh{x}$ satisfies $\angles{\wh{x}, r_1} \geq 1 - O(\delta)$ where $\delta \coloneqq \frac{1}{\log n}$, which means that $\|R(\wh{x} - r_1)\|_2 = \|R\wh{x} - e_1\|_2 \leq O(\delta)$.
    Since $v$ is a $(\rho,\gamma)$-compressible unit vector (for $\gamma$ bounded away from $1$), by \ref{item:close-to-e1-implies-sparse} of \Cref{lem:selecting-the-vector}, it follows that $\wh{v} = \wt{A} \wh{x}$ satisfies $\max_{S: |S|\leq \rho n} \|\wh{v}_S\|_2 \geq 1 - \gamma - O(\delta)$.

    Thus, in step (3) of \Cref{alg:planted-sparse-vector}, we will choose an $\wh{x}^*$ from the list such that $\wh{v}^* = \wt{A} \wh{x}^*$ satisfies $\max_{S: |S|\leq \rho n} \|\wh{v}_S\|_2 \geq 1 - \gamma - O(\delta)$.
    By \ref{item:far-from-e1-implies-spread} of \Cref{lem:selecting-the-vector}, we have that either $\|\wh{x}^* - r_1\|_2$ or $\|\wh{x}^* + r_1\|_2 \leq o(1)$.
    Since the singular values of $\wt{A}$ are all $1 \pm o(1)$ by \Cref{lem:A-norm}, it follows that $\wh{v}^*$ is $o(1)$-close to $\pm v$ and $\|\wh{v}^*\|_2 = 1 \pm o(1)$.
    This completes the proof.
\end{proof}

\section*{Acknowledgements}
We would like to thank Jeff Xu for inspiring technical discussions on the encoding scheme and PUR factors, Pravesh K.\ Kothari and Sidhanth Mohanty for discussions on this and related problems, and last but not least Hongjie Chen and Tommaso d'Orsi for discussions on their previous works \cite{dOrsiKNS20,ChendO22}.

\bibliographystyle{alpha}
\bibliography{main}

\newcommand{\etalchar}[1]{$^{#1}$}
\begin{thebibliography}{DKWB23}

\bibitem[ALPT11]{AdamczakLPT11}
Rados{\l}aw Adamczak, Alexander~E Litvak, Alain Pajor, and Nicole {Tomczak-Jaegermann}.
\newblock {Sharp bounds on the rate of convergence of the empirical covariance matrix}.
\newblock {\em Comptes Rendus. Math{\'e}matique}, 349(3-4):195--200, 2011.

\bibitem[AW08]{AminiW08}
Arash~A Amini and Martin~J Wainwright.
\newblock {High-dimensional analysis of semidefinite relaxations for sparse principal components}.
\newblock In {\em 2008 IEEE international symposium on information theory}, pages 2454--2458. IEEE, 2008.

\bibitem[BBH{\etalchar{+}}12]{BarakBHKSZ12}
Boaz Barak, Fernando~GSL Brandao, Aram~W Harrow, Jonathan Kelner, David Steurer, and Yuan Zhou.
\newblock {Hypercontractivity, Sum-of-Squares Proofs, and their Applications}.
\newblock In {\em Proceedings of the forty-fourth annual ACM symposium on Theory of computing}, pages 307--326, 2012.

\bibitem[BGG{\etalchar{+}}17]{BhattiproluGGLT17}
Vijay Bhattiprolu, Mrinalkanti Ghosh, Venkatesan Guruswami, Euiwoong Lee, and Madhur Tulsiani.
\newblock {Weak Decoupling, Polynomial Folds, and Approximate Optimization over the Sphere}.
\newblock In {\em 2017 IEEE 58th Annual Symposium on Foundations of Computer Science (FOCS)}, pages 1008--1019. IEEE, 2017.

\bibitem[BGL17]{BGL-random17}
Vijay Bhattiprolu, Venkatesan Guruswami, and Euiwoong Lee.
\newblock Sum-of-squares certificates for maxima of random tensors on the sphere.
\newblock In {\em Approximation, Randomization, and Combinatorial Optimization. Algorithms and Techniques, {APPROX/RANDOM}}, volume~81 of {\em LIPIcs}, pages 31:1--31:20. Schloss Dagstuhl - Leibniz-Zentrum f{\"{u}}r Informatik, 2017.

\bibitem[BKS14]{BarakKS14}
Boaz Barak, Jonathan~A Kelner, and David Steurer.
\newblock {Rounding Sum-of-Squares Relaxations}.
\newblock In {\em Proceedings of the forty-sixth annual ACM symposium on Theory of computing}, pages 31--40, 2014.

\bibitem[BKS15]{BarakKS15}
Boaz Barak, Jonathan~A Kelner, and David Steurer.
\newblock {Dictionary Learning and Tensor Decomposition via the Sum-of-Squares Method}.
\newblock In {\em Proceedings of the forty-seventh annual ACM symposium on Theory of computing}, pages 143--151, 2015.

\bibitem[BS16]{BarakS16}
Boaz Barak and David Steurer.
\newblock {Proofs, beliefs, and algorithms through the lens of sum-of-squares}.
\newblock {\em Course notes: \url{http://www.sumofsquares.org/public/index.html}}, 2016.

\bibitem[Cd22]{ChendO22}
Hongjie Chen and Tommaso d'Orsi.
\newblock {On the well-spread property and its relation to linear regression}.
\newblock In {\em Conference on Learning Theory}, pages 3905--3935. PMLR, 2022.

\bibitem[CT05]{CandesT05}
Emmanuel~J Candes and Terence Tao.
\newblock {Decoding by linear programming}.
\newblock {\em IEEE transactions on information theory}, 51(12):4203--4215, 2005.

\bibitem[DH13]{DemanetH13}
Laurent Demanet and Paul Hand.
\newblock {Recovering the Sparsest Element in a Subspace}.
\newblock {\em arXiv preprint arXiv:1310.1654}, 2013.

\bibitem[DK22]{DiakonikolasK22}
Ilias Diakonikolas and Daniel Kane.
\newblock {Non-Gaussian Component Analysis via Lattice Basis Reduction}.
\newblock In {\em Conference on Learning Theory}, pages 4535--4547. PMLR, 2022.

\bibitem[dKNS20]{dOrsiKNS20}
Tommaso d'Orsi, Pravesh~K Kothari, Gleb Novikov, and David Steurer.
\newblock {Sparse PCA: Algorithms, Adversarial Perturbations and Certificates}.
\newblock In {\em 2020 IEEE 61st Annual Symposium on Foundations of Computer Science (FOCS)}, pages 553--564. IEEE, 2020.

\bibitem[DKWB21]{DingKWB21}
Yunzi Ding, Dmitriy Kunisky, Alexander~S Wein, and Afonso~S Bandeira.
\newblock {The Average-Case Time Complexity of Certifying the Restricted Isometry Property}.
\newblock {\em IEEE Transactions on Information Theory}, 67(11):7355--7361, 2021.

\bibitem[DKWB23]{DingKWB23}
Yunzi Ding, Dmitriy Kunisky, Alexander~S Wein, and Afonso~S Bandeira.
\newblock {Subexponential-time algorithms for sparse PCA}.
\newblock {\em Foundations of Computational Mathematics}, pages 1--50, 2023.

\bibitem[dLN{\etalchar{+}}21]{dOrsiLNSS21}
Tommaso d'Orsi, Chih-Hung Liu, Rajai Nasser, Gleb Novikov, David Steurer, and Stefan Tiegel.
\newblock {Consistent Estimation for PCA and Sparse Regression with Oblivious Outliers}.
\newblock {\em Advances in Neural Information Processing Systems}, 34:25427--25438, 2021.

\bibitem[Don06]{Donoho06}
David~L Donoho.
\newblock {Compressed Sensing}.
\newblock {\em IEEE Transactions on information theory}, 52(4):1289--1306, 2006.

\bibitem[FKP19]{FlemingKP19}
Noah Fleming, Pravesh Kothari, and Toniann Pitassi.
\newblock {Semialgebraic Proofs and Efficient Algorithm Design}.
\newblock {\em Foundations and Trends{\textregistered} in Theoretical Computer Science}, 14(1-2):1--221, 2019.

\bibitem[FLM77]{FigielLM77}
T~Figiel, J~Lindenstrauss, and VD~Milman.
\newblock {The dimension of almost spherical sections of convex bodies}.
\newblock {\em Acta Mathematica}, 139:53--94, 1977.

\bibitem[GG84]{GarnaevG84}
Andrei~Yurevich Garnaev and Efim~Davydovich Gluskin.
\newblock {The widths of a Euclidean ball}.
\newblock In {\em Doklady Akademii Nauk}, volume 277, pages 1048--1052. Russian Academy of Sciences, 1984.

\bibitem[GKM22]{GuruswamiKM22}
Venkatesan Guruswami, Pravesh~K Kothari, and Peter Manohar.
\newblock {Algorithms and certificates for Boolean CSP refutation: smoothed is no harder than random}.
\newblock In {\em Proceedings of the 54th Annual ACM SIGACT Symposium on Theory of Computing}, pages 678--689, 2022.

\bibitem[GLR10]{GuruswamiLR10}
Venkatesan Guruswami, James~R Lee, and Alexander Razborov.
\newblock {Almost Euclidean subspaces of $\ell_1^N$ via expander codes}.
\newblock {\em Combinatorica}, 30(1):47--68, 2010.

\bibitem[GLW08]{GuruswamiLW08}
Venkatesan Guruswami, James~R Lee, and Avi Wigderson.
\newblock {Euclidean sections of $\ell_1^N$ with sublinear randomness and error-correction over the reals}.
\newblock In {\em International Workshop on Approximation Algorithms for Combinatorial Optimization}, pages 444--454. Springer, 2008.

\bibitem[HKM23]{HsiehKM23}
Jun-Ting Hsieh, Pravesh~K Kothari, and Sidhanth Mohanty.
\newblock {A simple and sharper proof of the hypergraph Moore bound}.
\newblock In {\em Proceedings of the 2023 Annual ACM-SIAM Symposium on Discrete Algorithms (SODA)}, pages 2324--2344. SIAM, 2023.

\bibitem[HKPT24]{HsiehKPT24}
Jun-Ting Hsieh, Pravesh~K Kothari, Lucas Pesenti, and Luca Trevisan.
\newblock {New SDP Roundings and Certifiable Approximation for Cubic Optimization}.
\newblock In {\em Proceedings of the 2024 Annual ACM-SIAM Symposium on Discrete Algorithms (SODA)}, pages 2337--2362. SIAM, 2024.

\bibitem[HKPX23]{HsiehKPX23}
Jun-Ting Hsieh, Pravesh~K Kothari, Aaron Potechin, and Jeff Xu.
\newblock {Ellipsoid Fitting up to a Constant}.
\newblock In {\em 50th International Colloquium on Automata, Languages, and Programming (ICALP 2023)}. Schloss-Dagstuhl-Leibniz Zentrum f{\"u}r Informatik, 2023.

\bibitem[Hop18]{Hopkins18}
Samuel Hopkins.
\newblock {Statistical inference and the sum of squares method}.
\newblock {\em PhD thesis, Cornell University}, 2018.

\bibitem[HSSS16]{HopkinsSSS16}
Samuel~B Hopkins, Tselil Schramm, Jonathan Shi, and David Steurer.
\newblock {Fast spectral algorithms from sum-of-squares proofs: tensor decomposition and planted sparse vectors}.
\newblock In {\em Proceedings of the forty-eighth annual ACM symposium on Theory of Computing}, pages 178--191, 2016.

\bibitem[Ind06]{Indyk06}
Piotr Indyk.
\newblock {Stable Distributions, Pseudorandom Generators, Embeddings, and Data Stream Computation}.
\newblock {\em Journal of the ACM (JACM)}, 53(3):307--323, 2006.

\bibitem[Ind07]{Indyk07}
Piotr Indyk.
\newblock {Uncertainty principles, extractors, and explicit embeddings of $\ell_2$ into $\ell_1$}.
\newblock In {\em Proceedings of the thirty-ninth annual ACM symposium on Theory of computing}, pages 615--620, 2007.

\bibitem[IS10]{IndykS10}
Piotr Indyk and Stanislaw Szarek.
\newblock {Almost-Euclidean Subspaces of via Tensor Products: A Simple Approach to Randomness Reduction}.
\newblock In {\em International Workshop on Randomization and Approximation Techniques in Computer Science}, pages 632--641. Springer, 2010.

\bibitem[JL09]{JohnstoneL09}
Iain~M Johnstone and Arthur~Yu Lu.
\newblock {On consistency and sparsity for principal components analysis in high dimensions}.
\newblock {\em Journal of the American Statistical Association}, 104(486):682--693, 2009.

\bibitem[JPR{\etalchar{+}}22]{JonesPRTX22}
Chris Jones, Aaron Potechin, Goutham Rajendran, Madhur Tulsiani, and Jeff Xu.
\newblock {Sum-Of-Squares Lower Bounds for Sparse Independent Set}.
\newblock In {\em 2021 IEEE 62nd Annual Symposium on Foundations of Computer Science (FOCS)}, pages 406--416. IEEE, 2022.

\bibitem[Kas77]{Kashin77}
Boris~Sergeevich Kashin.
\newblock {Diameters of some finite-dimensional sets and classes of smooth functions}.
\newblock {\em Izvestiya Rossiiskoi Akademii Nauk. Seriya Matematicheskaya}, 41(2):334--351, 1977.

\bibitem[KT07]{KashinT07}
Boris~S Kashin and Vladimir~N Temlyakov.
\newblock {A Remark on Compressed Sensing}.
\newblock {\em Mathematical notes}, 82:748--755, 2007.

\bibitem[KWB19]{KuniskyWB19}
Dmitriy Kunisky, Alexander~S Wein, and Afonso~S Bandeira.
\newblock {Notes on computational hardness of hypothesis testing: Predictions using the low-degree likelihood ratio}.
\newblock {\em arXiv preprint arXiv:1907.11636}, 2019.

\bibitem[KZ14]{KoiranZ14}
Pascal Koiran and Anastasios Zouzias.
\newblock {Hidden Cliques and the Certification of the Restricted Isometry Property}.
\newblock {\em IEEE transactions on information theory}, 60(8):4999--5006, 2014.

\bibitem[LS08]{LovettS08}
Shachar Lovett and Sasha Sodin.
\newblock {Almost Euclidean sections of the $N$-dimensional cross-polytope using $O(N)$ random bits}.
\newblock {\em Communications in Contemporary Mathematics}, 10(04):477--489, 2008.

\bibitem[MW21]{MaoW21}
Cheng Mao and Alexander~S Wein.
\newblock {Optimal Spectral Recovery of a Planted Vector in a Subspace}.
\newblock {\em arXiv preprint arXiv:2105.15081}, 2021.

\bibitem[OZ13]{ODonnelW13}
Ryan O'Donnell and Yuan Zhou.
\newblock {Approximability and proof complexity}.
\newblock In {\em Proceedings of the twenty-fourth annual ACM-SIAM symposium on Discrete algorithms}, pages 1537--1556. SIAM, 2013.

\bibitem[QSW14]{QuSW14}
Qing Qu, Ju~Sun, and John Wright.
\newblock {Finding a sparse vector in a subspace: Linear sparsity using alternating directions}.
\newblock {\em Advances in Neural Information Processing Systems}, 27, 2014.

\bibitem[RRS17]{RaghavendraRS17}
Prasad Raghavendra, Satish Rao, and Tselil Schramm.
\newblock {Strongly refuting random CSPs below the spectral threshold}.
\newblock In {\em Proceedings of the 49th Annual {ACM} {SIGACT} Symposium on Theory of Computing, {STOC} 2017, Montreal, QC, Canada, June 19-23, 2017}, pages 121--131. {ACM}, 2017.

\bibitem[SWW12]{SpielmanWW12}
Daniel~A Spielman, Huan Wang, and John Wright.
\newblock Exact recovery of sparsely-used dictionaries.
\newblock In {\em Conference on Learning Theory}, pages 37--1. JMLR Workshop and Conference Proceedings, 2012.

\bibitem[Tao12]{Tao12}
Terence Tao.
\newblock Topics in random matrix theory.
\newblock {\em Graduate Studies in Mathematics}, 2012.

\bibitem[Ver20]{Ver20}
Roman Vershynin.
\newblock {High-Dimensional Probability}.
\newblock {\em University of California, Irvine}, 2020.

\bibitem[WAM19]{WeinAM19}
Alexander~S Wein, Ahmed~El Alaoui, and Cristopher Moore.
\newblock {The Kikuchi hierarchy and tensor PCA}.
\newblock In {\em 2019 IEEE 60th Annual Symposium on Foundations of Computer Science (FOCS)}, pages 1446--1468. IEEE, 2019.

\bibitem[ZSWB22]{ZadikSWB22}
Ilias Zadik, Min~Jae Song, Alexander~S Wein, and Joan Bruna.
\newblock {Lattice-Based Methods Surpass Sum-of-Squares in Clustering}.
\newblock In {\em Conference on Learning Theory}, pages 1247--1248. PMLR, 2022.

\end{thebibliography}

\end{document}